\documentclass{article}

\usepackage{arxiv}

\usepackage{multirow}%
\usepackage{amsmath,amssymb,amsfonts}%
\usepackage{amsthm}%
\usepackage{mathrsfs}%
\usepackage[title]{appendix}%
\usepackage{xcolor}%
\usepackage{textcomp}%
\usepackage{manyfoot}%
\usepackage{booktabs}%
\usepackage{algorithm}%
\usepackage{algorithmicx}%
\usepackage{algpseudocode}%
\usepackage{listings}%
\usepackage{array}
\usepackage{enumerate}
\usepackage{verbatim}
\usepackage{epsfig}
\usepackage{hyperref}       
\usepackage{url} 

\usepackage[dvipsnames]{xcolor}

\hypersetup{
    colorlinks=true,
    linkcolor=Maroon,       
    citecolor=NavyBlue,     
    urlcolor=DarkOrchid     
}

\usepackage{appendix}

\usepackage{graphicx}
\usepackage{natbib}
\usepackage{booktabs}
\usepackage{mathbbol}
\usepackage{ulem}

\allowdisplaybreaks

\newcommand{\mb}{\mathbb}
\newcommand{\mc}{\mathcal}
\newcommand{\bx}{{\bf x}}
\newcommand{\bX}{{\bf X}}
\newcommand{\bS}{{\bf S}}

\newcommand{\bU}{{\bf U}}

\newcommand{\bV}{{\bf V}}
\newcommand{\ba}{{\bf a}}

\newcommand{\bY}{{\bf Y}}
\newcommand{\by}{{\bf y}}

\newcommand{\bA}{{\bf A}}

\newcommand{\bG}{{\bf G}}

\newcommand{\bD}{{\bf D}}

\newcommand{\bp}{{\bf p}}

\def\rT{\mathrm T}

\newcommand\independent{\protect\mathpalette{\protect\independenT}{\perp}}
\def\independenT#1#2{\mathrel{\rlap{$#1#2$}\mkern2mu{#1#2}}}

\newcommand{\beq}{\begin{equation}}
\newcommand{\eeq}{\end{equation}}
\newcommand{\beqn}{\begin{eqnarray}}
\newcommand{\eeqn}{\end{eqnarray}}
\newcommand{\bali}{\begin{align*}}
\newcommand{\eali}{\end{align*}}

\newcommand{\bUpi}{{\boldsymbol \pi}}

\newcommand{\bUGamma}{{\boldsymbol \Gamma}}

\newcommand{\bUS}{{\boldsymbol \Sigma}}

\newcommand{\eps}{{\varepsilon}}

\def\h2s{{\hspace{0.2in}}}
\def\b1s{{\hspace{-0.1in}}}

\theoremstyle{plain}%
\newtheorem{theorem}{Theorem}
\theoremstyle{remark}%
\newtheorem{remark}{Remark}%

\newtheorem{lemma}{Lemma}
\newtheorem{coro}{Corollary}

\theoremstyle{definition}%

\title{Evaluating Treatment Effects using Group Testing with Retesting of Positive Groups}

\author{ \href{https://orcid.org/0009-0006-0668-1304}{\includegraphics[scale=0.06]{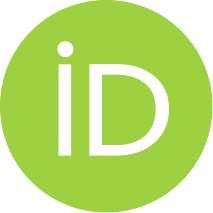}}\hspace{1mm}Aye Aye~Maung \\
	Department of Bioinformatics and Biostatistics\\
	University of Louisville\\
	Louisville, KY, USA 40202 \\
	\texttt{ayeaye.maung@louisville.edu} \\
	\And
	Qi~Zheng \\
	Department of Bioinformatics and Biostatistics\\
	University of Louisville\\
	Louisville, KY, USA 40202 \\
	\texttt{qi.zheng@louisville.edu} \\
}

\renewcommand{\shorttitle}{Evaluating Treatment Effects using Group Testing with Retesting of Positive Groups}

\hypersetup{
pdftitle={Evaluating Treatment Effects using Group Testing with Retesting of Positive Groups},
pdfauthor={Aye Aye~Maung, Qi~Zheng},
pdfkeywords={Pooled screening, Observational studies, Inverse probability weighting, Confounding adjustment, Pseudo-likelihood estimation},
}

\begin{document}

\maketitle


\begin{abstract}
Group testing is an established, highly cost-effective strategy for population-level disease surveillance that identifies positive individuals by pooling biological specimens. Originally introduced during World War II for large-scale screening and heavily utilized in modern high-throughput public health infrastructure, traditional group testing methods are restricted to purely associational analyses. Consequently, they lack the capacity to infer the direct causal effect of an intervention when individual-level data are subject to baseline confounding. In this work, we bridge this fundamental gap by introducing a causal inference framework tailored specifically for group testing designs. We integrate the principles of inverse probability weighting (IPW) directly into a pooled pseudo-likelihood formulation to construct an unbiased pseudo-score function. Under standard regularity conditions, we prove the consistency and asymptotic normality of our proposed plug-in estimator. Extensive numerical simulations demonstrate that our framework successfully purges severe selection bias, accurately recovering the true average treatment effect where traditional unweighted pooling models fail. Finally, we illustrate the practical utility of our method as both an estimation and a diagnostic tool using real-world observational surveillance data from the CDC's U.S. Influenza Vaccine Effectiveness Network.
\end{abstract}

%


\section{Introduction}
\label{sec1}

Group testing, initially proposed by \citet{dorfman1943detection} to screen for syphilis among World War II military recruits, has become a powerful technique to screen for various infectious diseases, including Chlamydia trachomatis and Neisseria gonorrhoeae \citep{lindan2005utility}, influenza \citep{van2012pooling}, herpes \citep{hill2016efficient}, malaria \citep{hsiang2012surveillance}, and COVID-19 \citep{brault2021group}. By pooling samples, group testing significantly expands testing capacity, allowing health departments to evaluate large populations with fewer diagnostic kits \citep{10.3389/fpubh.2021.583377}. This advantage is critical during the early stages of a pandemic when massive gaps exist between sharp increases in testing demand and supply shortages. Consequently, pooling has been widely implemented across multiple countries, including Australia \citep{chong2020sample}, Israel \citep{yelin2020evaluation}, South Korea \citep{kim2020pooling}, and the U.S. \citep{abdalhamid2020assessment}.

Existing statistical research in group testing mainly focuses on two primary areas: (1) improving diagnostic accuracy and optimizing cost-saving designs \citep[e.g.,][]{kim2007comparison, bilder2010optimal, tebbs2013nonparametric}, and (2) identifying subject-level features through regression modeling \citep[e.g.,][]{wang2013semi, lin2019regression}. However, a significant methodological gap exists: it remains unknown whether population-level causal treatment effects can be accurately assessed using group testing data. This gap is highly consequential because group testing is most valuable during large-scale outbreaks, precisely when public health interventions—such as mandates or new treatments—are concurrently implemented. Quantifying the causal effects of these measures is essential to avoid confounding bias and optimize resource allocation \citep{Glass2013-ns, Al-Kassab-Cordova2025-mp}, making the development of statistical methods for evaluating intervention effects from pooled data a critical area for pandemic preparedness.

Evaluating causal effects from group testing data is exceptionally challenging due to the inherent loss of individual-level information. Accurate causal estimation requires adjusting for individual-level confounding variables (e.g., age, health status), yet standard statistical tools like inverse probability weighting (IPW) fail ``off the shelf" because pooling severs the essential link between an individual's outcome and their specific confounders. Overcoming this information loss requires custom-built models to handle the complex, non-linear relationships induced by the pooling architecture. 

This paper bridges this methodological gap by introducing a unified, likelihood-based framework for estimating the average treatment effect (ATE) from group testing data. Our approach handles data from both randomized clinical trials and observational studies. For randomized data, we derive the exact likelihood. For observational data, we develop a novel pseudo-likelihood function to integrate IPWs directly into the group testing likelihood structure. We provide rigorous theoretical justification, proving that the resulting ATE estimator is consistent and asymptotically normal under standard regularity conditions. Finally, we demonstrate the robustness of our method through extensive simulation studies and a case study on real-life influenza vaccination data, establishing that group testing can serve as a robust framework for correctly assessing public health interventions.

The remainder of this paper is organized as follows. Section~\ref{GT_Data} outlines the mathematical notation of the pooled testing structure. Section~\ref{GT_RCM} establishes the causal framework under the potential outcomes setup. Section~\ref{GT_PseudoLL} introduces our weighted pseudo-likelihood estimator and its corresponding asymptotic properties. Section~\ref{simulation_set} and Section~\ref{simulation_res} detail the numerical simulation configurations and performance results. Section~\ref{casestudy} presents the empirical case study application, and Section~\ref{discussion} concludes with a brief discussion.
	
	\section{Methodology} \label{methodology}
	\subsection{Group Testing Data} \label{GT_Data}
	Let $\bX\in \mc{X}$ and $A\in \mc{A}$ denote the $p$-dimensional vector of baseline covariates ($p \in \mb{N}$) and the categorical treatment received by the subject, respectively.
	The support spaces $\mc{X}$ and $\mc{A}$ represent a compact set in $\mb{R}^{p}$ and a finite set, respectively.  Let $Y \in \mc{Y} = \{0, 1\}$ denote the latent, binary true infection status of the subject.

    In this paper, we specifically deal with the group testing framework which incorporates retesting of individuals in positive groups. We aggregate $N \in \mb{N}$ independent individuals into a single pool. The vector of individual true outcomes, $\bY =  \{Y_j\}_{j = 1}^N \in \mc{Y}^N$, is unobserved. If the diagnostic assay were perfect, we would observe the true aggregated pool status, $Y^{*}=\max_{j= 1}^N Y_{j}$. However, due to potential diagnostic misclassification, we instead observe an imperfect binary pool test result $\widetilde{Y}^{\ast} \in \mc{Y}$.  It is worth noting that in this paper, we use a tilde ($\widetilde{\cdot}$)  to denote proxy or estimated values, such as statuses observed from imperfect testing kits or parameter estimates. To incorporate individual retesting when a pool tests positive, every individual within the group is retested if and only if $\widetilde{Y}^{\ast} = 1$. The resulting individual test outcomes are recorded in the vector $\widetilde{\bY} = \{\widetilde{Y}_j\}_{j = 1}^N \in \mc{Y}^N$.

    Let $S_e = P(\widetilde{Y}^{\ast}=1 \mid Y^{*}=1)$ and $S_p = P(\widetilde{Y}^{\ast}=0 \mid Y^{*}=0)$ denote the sensitivity and specificity of the pool-level assay, respectively. We assume that individual-level diagnostic accuracy matches the pool-level parameters, such that $P(\widetilde{Y}_j = 1 \mid Y_j = 1) = S_e$ and $P(\widetilde{Y}_j = 0 \mid Y_j = 0) = S_p$ for all $j$. That is, we assume there is no dilution effect of pooled testing. This is a reasonable assumption for low to moderate pool sizes \citep{Bateman2020-zs}  or for novel pooling procedures that mitigate sample dilution \citep{Christoff2021-zc}.
    The complete outcome of this multi-stage group test-retest process can be compactly represented by a single random variable
    \begin{align*}
    \widetilde{Z} = \left\{\begin{array}{cc} 
    0 &\text{if } \widetilde{Y}^{\ast} = 0 \\ 
    \widetilde{\bY}  &\text{if } \widetilde{Y}^{\ast} = 1\end{array}
    \right.
    \end{align*}defined on the support space $\mc{Z} = \{0\} \bigcup \mc{Y}^N$.

    Consequently, the observed data vector for a given pool is denoted by $D = (\widetilde{Z}, \bA, \mb{X})$, where $\mb{X}=(\bX_{1}^{\rT},\cdots,\bX_{N}^{\rT})^{\rT}$ represents the vector of individual covariates and $\bA=(A_{1},\cdots, A_{N})^{\rT}$ represents the vector of individual treatments within the pool. The observed dataset, $\bD = \{D_i\}_{i=1}^{M}$, consists of $M \in \mb{N}$ independent and identically distributed (i.i.d.) replicates of $D$. The total number of subjects in the study is  $M_{t}=\sum_{i=1}^{M}N_i$. 
    
	
	\subsection{Treatment Evaluation Framework for Group Tests} \label{GT_RCM}

	Group testing is primarily deployed for rapid population surveillance and disease screening. Consequently, group testing studies are fundamentally observational, meaning that the random assignment of subjects to treatment groups is rarely guaranteed. Isolating subject-level treatment effects from such data requires formal adjustments for baseline confounders to eliminate selection bias \citep{10.1093/aje/kwn010, Igelstrom960}.
    
To achieve identifiability of causal treatment effects, we employ the Rubin Causal Model (RCM) framework \citep{hernan2020causal}. At the individual level, let $Y^{(a)}$ denote the potential binary infection outcome of a subject, had they received treatment level $a \in \mc{A}$. We assume these potential outcomes satisfy the standard structural axioms:
	\begin{enumerate}[({RCM}1)]
		\item \label{RCM1} \textit{Consistency}: The observed outcome, under each treatment level, equals the potential outcome corresponding to that level: $Y = \sum_{a \in \mc{A}} 1\{A = a\} Y^{(a)}$
		\item \label{RCM2} \textit{Conditional Exchangeability}: The potential outcomes are independent of treatment conditional on the covariates: $\{Y^{(a)}\}_{a \in \mc{A}} \independent A \ | \ \bX$.
	\end{enumerate}

    To translate this framework to the group-testing structure, let $\widetilde{Z}^{(\ba)}$ define the potential outcome of the joint test-retest process for a pool, had the constituent individuals received the treatment vector $\ba = \{a_j\}_{j = 1}^N \in \mc{A}^N$. We stipulate the group-level potential outcomes to satisfy analogous consistency and exchangeability conditions: $\widetilde{Z} = \sum_{\ba \in \mc{A}^N} 1\{\bA = \ba\}\widetilde{Z}^{(\ba)}$ and $\{\widetilde{Z}^{(\ba)}\}_{\ba \in \mc{A}^N} \independent \bA \ | \ \mb{X}$.
	

    Without loss of generality, we restrict our focus to a binary treatment space $\mc{A} = \{0, 1\}$, where $0$ and $1$ denote the control and active treatment conditions, respectively. Because our core public health objective is to evaluate the population-level efficacy of an intervention on reducing overall infection rates, we target the average treatment effect (ATE), defined as $\tau = E[Y^{(1)}] - E[Y^{(0)}]$. Let $p_a^{*} = E[Y^{(a)}]$ for $a \in \{0, 1\}$ represent the marginal potential outcome probabilities, and denote the parameter vector of interest as $\mathbf{p}^{*} = (p_0^{*}, p_1^{*})^{\rT}$.
    
    Because our core objective is the consistent estimation of the  parameter vector $\mathbf{p}^{*}$ rather than individual clinical diagnosis or updating misclassification rates dynamically \citep{gastwirth1994screening, krajden2014pooled, mcmahan2017bayesian}, we assume that the sensitivity ($S_{e}$) and specificity ($S_{p}$) of the assay are known or pre-estimated. Our framework is oriented strictly towards maximizing statistical efficiency during parameter estimation.

\subsection{Pseudo-Likelihood for Estimating Treatment Effects} \label{GT_PseudoLL}

We first consider the estimation of the ATE using group testing data collected from randomized controlled trials (RCT). Randomized controlled trials are widely regarded as the gold standard for evaluating causal treatment effects because subjects are randomly assigned to different treatment arms. Consequently, the treatment $A$ is independent of all baseline covariates, whether measured or unmeasured. This balances the distribution of confounders across treatment groups, meaning that unconditional exchangeability applies: $Y^{(a)} \independent A$.

Under the consistency condition and this stronger form of exchangeability condition,
\begin{align*}
    E\left[Y|A=a\right]=E\left[Y^{(a)}|A=a\right]=E[Y^{(a)}]=p_a^{*}.
\end{align*}
This relationship implies that if individual-level binary outcomes were directly available, the marginal potential outcome probabilities $p^{*}_{a}$ ($a=0,1$) could be consistently estimated by the conditional sample average within each treatment arm.

Since we lack direct access to individual outcomes, in a group testing framework, we must construct a likelihood function based entirely on the observed pool-level data $(\widetilde{Z}, \bA)$ and the parameters. To this end, let $\bp^{*}=(p_0^{*}, p_1^{*})^{\rT}$ be the true parameters and let $\bp=(p_0, p_1)^{\rT}$ be parameter variables. Let $\ba$ be a realized treatment vector of the pool of the $N$ subjects. By strong exchangeability, the distribution of the potential pooled outcome $\widetilde{Z}^{(\ba)}$ and the conditional distribution of the observed pooled outcome $\widetilde{Z}$ given $\ba$ coincide: $\pi_{\widetilde{Z}^{(\ba)}}(\tilde{z}; \bp) = \pi_{\widetilde{Z} | \bA = \ba}(\tilde{z}; \bp)$. Hence the function
\begin{align}
    \ell(\bp| \widetilde{Z}, \bA) := &\log \pi_{\widetilde{Z}^{(\ba)}}(\widetilde{Z}; \bp) \label{eq:individual_po_ll} \\
= &\log \pi_{\widetilde{Z} | \bA = \ba}(\widetilde{Z}; \bp)\label{individual_cond_ll}
\end{align} constitutes a valid log-likelihood of $(\widetilde{Z}, \bA)$. So, based on the data $\bD^{(0)}:= \{(\widetilde{Z}_i, \bA_i)\}_{i = 1}^M$, the log-likelihood is
\begin{align}
    &\ell(\bp|\bD^{(0)}) := \frac{1}{M}\sum_{i=1}^{M} \ell(\bp| \widetilde{Z}_i, \bA_i) 
    = \frac{1}{M}\sum_{i = 1}^M \sum_{\ba_i \in \mc{A}^{N_i}} 1\{\bA_i = \ba_i\} \log \pi_{\widetilde{Z}_i^{(\ba_i)}}(\widetilde{Z}_i; \bp) \label{eq:original_ll} \\
    = &\frac{1}{M}\sum_{i = 1}^M \sum_{\ba_i \in \mc{A}^{N_i}} 1\{\bA_i = \ba_i\} \log \pi_{\widetilde{Z}_i | \bA_i = \ba_i}(\widetilde{Z}_i; \bp). \label{RCT_ll}
\end{align}

In Appendix \ref{rct_likelihood}, we expand the log-likelihood \eqref{RCT_ll} fully in terms of the observed data and parameters. There, $\widetilde{\eta}_0(\ba, \bp)$ denotes
the algebraic form, under RCT conditions, of $P[\widetilde{Y}^{*}=0|\bA=\ba; \bp]$, the conditional probability of observing a negative group test result. Similarly, $
\widetilde{\phi}_1(\widetilde{\by}, \ba; \bp)$ denotes the algebraic form, under RCT conditions, of $P[\widetilde{Y}^{*}=1, \widetilde{\bY} = \widetilde{\by}|\bA = \ba; \bp]$, the conditional probability of observing a positive group test followed by individual-specific retest outcomes $\widetilde{\by}$.  So the RCT log-likelihood takes the form
\begin{align}
    \ell(\bp|\bD^{(0)}) = \frac{1}{M}\sum_{i = 1}^M \Big\{ (1-\widetilde{Y}^{*}_{i}) \log \widetilde{\eta}_0(\bA_i;\bp) +\widetilde{Y}^{*}_i \log\widetilde{\phi}_1(\widetilde{\bY}_i, \bA_i; \bp) \Big\}. \label{RCT_ll_expanded}
\end{align} Note that for RCTs, there is no need to incorporate baseline covariates into the log-likelihood. Simply optimizing \eqref{RCT_ll_expanded} provides valid, consistent estimates for $\bp^\ast$.

In experimental designs, investigators actively control and randomize treatment assignments to isolate causal effects; in observational studies, however, participants are evaluated without any experimental intervention. The overarching objective of these observational studies is to investigate the natural distribution and progression of risk factors, exposures, and epidemiological outcomes \citep{Gilmartin-Thomas2018-tk}. Within group testing settings, observational frameworks are often essential because investigators cannot control treatment assignments. Instead, data are typically inherited from routine, population-level screening databases where sample pooling and individual exposures are determined by real-world logistics rather than experimental protocols \citep{lindan2005utility, van2012pooling}. Furthermore, ethical constraints frequently prohibit the active assignment of harmful exposures, such as smoking status or environmental toxins \citep{Goldstein2018-ih}. In addition to bypassing these ethical hurdles, observational studies are typically faster and more cost-effective to implement than large-scale randomized trials \citep{Faraoni2016-fx}. Driven by modern advancements in high-throughput data collection, these observational datasets have become widely available, offering a powerful avenue for estimating the Average Treatment Effect (ATE) from routine screening databases. Consequently, the vast majority of population-level group testing surveillance data remains observational in nature \citep[see, e.g.,][]{lindan2005utility, van2012pooling, hsiang2012surveillance}.

In an observational study, the underlying treatment assignment mechanism is typically unknown and often depends on subject-specific characteristics or clinical preferences. Consequently, baseline covariates may be highly unbalanced across treatment arms, meaning that differences in observed infection rates capture a mixture of true treatment effects and confounding bias. Under the conditional exchangeability assumption outlined in Section~\ref{GT_RCM}, we have
\begin{align*}
    E[E[Y|A = a, \bX]] = E[Y^{(a)}] = p_a^\ast
\end{align*} implying that $p_a^\ast$ can no longer be consistently estimated by the unadjusted conditional sample averages, as $E[Y |A = a] \neq p_a^\ast$ in the presence of confounding. As a result, $\pi_{\widetilde{Z}^{(\ba)}}(\tilde{z}; \bp)$, $\pi_{\widetilde{Z} | \bA = \ba}(\tilde{z}; \bp)$ diverge and the specific equation lines \eqref{individual_cond_ll} and \eqref{RCT_ll} fail.

Furthermore, the algebraic forms $\widetilde{\eta}_0(\ba;\bp)$ and $\widetilde{\phi}_1(\widetilde{\by}, \ba; \bp)$, originally derived from conditional outcome distributions under RCT conditions in Appendix \ref{rct_likelihood}, have modified interpretations in observational studies. As shown in Appendix \ref{os_derive}, they are now merely related to the potential outcome distributions via  
\begin{align*}
    \pi_{\widetilde{Z}^{(\ba)}}(\tilde{z}; \bp) = \Big[\widetilde{\eta}_0(\ba;\bp)\Big]^{1-\widetilde{y}^{*}}\Big[\widetilde{\phi}_1(\widetilde{\by}, \ba; \bp)\Big]^{\widetilde{y}^{*}}.
\end{align*} This alone is insufficient to make the function in Equation \eqref{RCT_ll_expanded} a true log-likelihood of the data. Indeed, equation line \eqref{eq:original_ll} shows that the function evaluates each observed pooled outcome $\widetilde{Z}_i$ under the probability density of the potential outcome variable $\widetilde{Z}_i^{(\ba_i)}$. And, unlike the situation in RCTs, since the conditional outcome distributions systematically diverge from the potential outcome distributions in the presence of confounding, optimizing the function yield consistently biased parameters.

To overcome this systemic bias, we leverage the seminal propensity score framework introduced by \citet{rosenbaum1983central}. By defining the conditional probability of treatment assignment given observed baseline covariates, propensity scores serve as a powerful balancing mechanism in observational research. Conditioning on the propensity score aligns the distribution of measured baseline covariates across treatment arms, ensuring that any remaining differences in the outcome distributions across treatments are uniquely due to the true causal effect \citep{hernan2020causal}. While various propensity score architectures exist, including matching, stratification, and direct covariate adjustment, this paper establishes a novel integration of inverse probability weighting (IPW) directly within the group testing framework.

To this end, the pool-level inverse probability weight for a given group data vector $D = (\widetilde{Z}, \bA, \mb{X})$ is formulated as
\begin{align}
W(\bA, \mb{X})= \sum_{\ba\in\mathcal{A}^{N} } \frac{1\{\bA=\ba\} }{\pi_{\bA}(\ba|\mb{X}) } 
= \prod_{j=1}^{N}\left(
\sum_{a=0}^{1} \frac{ 1\{A_{j}=a\} }{\pi_{a}^{\ast } (\bX_{j})}\right) \label{GT_IPW}
\end{align} where $\pi_{a}^{\ast}(\bx)=P(A=a|\bX=\bx)$ represents the individual-level propensity score. By weighting the contribution of each independent pool in \eqref{eq:original_ll}, we construct a novel pseudo-log-likelihood function that incorporates the baseline covariates:
\begin{align}
    &\ell(\bp|\bD) := \frac{1}{M} \sum_{i = 1}^M W(\bA_i, \mb{X}_i) \ell(\bp|\widetilde{Z}_i, \bA_i) \nonumber \\  = &\frac{1}{M} \sum_{i = 1}^M W(\bA_i, \mb{X}_i) \sum_{\ba_i \in \mc{A}^{N_i}}  1\{\bA_i = \ba_i\} \log \pi_{\widetilde{Z}_i^{(\ba_i)}}(\widetilde{Z}_i; \bp) \nonumber \\
    = &\frac{1}{M}\sum_{i = 1}^M W(\bA_i, \mb{X}_i) \Big\{ (1-\widetilde{Y}^{*}_{i}) \log \widetilde{\eta}_0(\bA_i;\bp) +\widetilde{Y}^{*}_i \log\widetilde{\phi}_1(\widetilde{\bY}_i, \bA_i; \bp) \Big\}. \label{eq:pseudo_ll}
\end{align} 

Although \eqref{eq:pseudo_ll} is not the exact log-likelihood for the complete observational data $\bD$, we prove that it preserves its essential properties; namely, setting the gradient of the pseudo-log-likelihood to zero yields unbiased estimating equations for the true parameter vector $\bp^\ast$. Let $\bU(\bp) := \bU(\bp; \bD) := \nabla \ell(\bp|\bD)$ denote the pseudo-score function. Under the technical regularity conditions detailed in Appendix \ref{theory}, we establish the following asymptotic properties for the pseudo-maximum likelihood estimator:\begin{theorem}\label{thm1}There exists a sequence of solutions $\tilde{\bp}$ to the pseudo-score equations $\bU(\bp)=0$ such that: (i) $\widetilde{\bp} \xrightarrow{p} \bp^{*}$;(ii) $\sqrt{M}\left(\widetilde{\bp}-\bp^{*}\right) \xrightarrow{d} N(0, \bV \bUS \bV^{\rT})$, where $\bV = -\{E[ \dot{\bU}(\bp^{*})]\}^{-1}$ and $\bUS = M E[\bU(\bp^{*})\bU(\bp^{*})^{\rT}]$.\end{theorem}

Note that despite its formal definition, the matrix $\bUS$ is functionally independent of the number of pools $M$ because $M E[\bU(\bp^{*})\bU(\bp^{*})^{\rT}] = E[\bU(\bp^{*}; D_1) \bU(\bp^{*}; D_1)^{\rT}]$, where $\bU(\bp^{*}; D_1) = W(\bA_1, \mb{X}_1) \nabla \ell(\bp^{*}|\widetilde{Z}_1, \bA_1)$ (cf. the proof of Lemma 4 in Appendix \ref{theory}).

In practice, the true propensity scores are unknown and the weights must be consistently estimated from the data. Let $\widehat{W}(\bA_i, \mb{X}_i)=\prod_{j=1}^{N}\left(\sum_{a=0}^{1}1\{A_{ij}=a\}\widehat{\pi}_{a}^{-1}(\bX_{ij})\right)$ denote the estimated pool-level weights. We first establish the the consistency of the resulting estimator as long as the weight estimates are themselves consistent:
\begin{theorem}\label{cor1}
        If $\|\widehat{\pi}_{1}-\pi^{*}_{1}\|_{\infty}=o_{p}(1)$,
        under Conditions (A1) -- (A3) of Appendix \ref{theory},
        there exists solutions  $\widehat{\bp}$ of $\widehat{\bU}(\bp)=0$ satisfying $\widehat{\bp}\rightarrow_{p}\bp^{*}$.
    \end{theorem}

We next establish the asymptotic normality of the pseudo-maximum likelihood estimator using consistent estimated weights:
    
\begin{theorem}\label{thm2}
    If (1) $\widehat{\pi}_1(\cdot)-\pi_1^{*}(\cdot)=\mb{E}_{M_t}\phi(\bX_{ij},\cdot)+o_{p}(M_t^{-1/2})$ for some function $\phi$ satisfying 
    $E[\phi(\bX_{ij},\bx)]=0$  and $E[\phi^2(\bX_{ij},\bx)]$ is bounded above uniformly over $\bx\in\mc{X}$; and (2) $\widehat{\pi}_1\in \mc{H}$ and
    $\|\widehat{\pi}_1-\pi_1^{*}\|_{\infty}=o_{p}(M_t^{-1/4})$, where $\mc{H}$ be a vector space of functions with maximum norm bounded above by some constant and the covering number \citep[see, e.g., page 132 of ][]{vaart2023weak}, $\int_{0}^{\infty}\sqrt{\log N(\epsilon, \mc{H}, \|\cdot\|_{\infty})}d\epsilon<\infty$, then 
    $\sqrt{M}(\widehat{\bp}-\bp^{*})\rightarrow_{d} N(0, \bUGamma)$, for some matrix $\bUGamma$.
\end{theorem}

Note that while the proof in Appendix \ref{theory} rigorously establishes the existence of the asymptotic variance matrix $\bUGamma$, it does not provide the explicit algebraic form. In our numerical studies, variance quantities that derive from $\bUGamma$, such as the variance of ATE estimates, are approximated via Monte Carlo simulations and bootstrapping.

\section{Numerical Studies} \label{simulation}

\subsection{Settings} \label{simulation_set}

To validate the performance of our proposed inverse probability weighted pseudo-likelihood estimator on confounded data, we simulated parametric data where the baseline covariates $\bX$ systematically confound the effect of the treatment $A$ on the binary infection outcome $Y$. This confounding is achieved via a hierarchical generation scheme: first, the covariate vector $\bX$ is drawn from a specified distribution; second, individual treatment assignments are generated via a Bernoulli scheme conditioned on the covariates, $A \mid \bX = \bx \sim \text{Bern}(p_{\bx})$; finally, potential individual infection statuses are simulated as $Y^{(a)} \mid \bX = \bx \sim \text{Bern}(p_{\bx, a})$, $a = 0, 1$, while the observed individual infection statuses are set to $Y = \sum_{a=0,1} 1\{A = a\} Y^{(a)}$. Here, $p_{\bx}$ and $p_{\bx, a}$ denote the underlying success probability parameters governing treatment assignment and disease development, respectively. We evaluated our method under two distinct structural settings:
\begin{enumerate}[({M}1)]
    \item The first setting is explicitly inspired by real-world influenza vaccination surveillance data reported by \citet{USFluVE2018}. It features three baseline covariates, $\bX = (X_1, X_2, X_3)^{\rT}$, heuristically representing patient age, underlying health status, and gender, respectively. These attributes are generated independently according to:
    $$
    X_1 \sim N(40, 20^2), \ X_2 \sim \text{Bern}(0.2), \ \text{and} \ X_3 \sim \text{Bern}(0.5).
    $$ 
    The individual treatment assignment and true infection status represent vaccination and influenza infection, respectively. The corresponding conditional probabilities are defined via the logistic parameters:
    $p_{\bx} = \text{logit}^{-1}(5 -\boldsymbol{\beta}_{1}^{\rT}\bx)$ and $p_{\bx, a} = \text{logit}^{-1}(7.2 -\boldsymbol{\beta}_{2}^{\rT}\bx +2a)
    $,
    where $\boldsymbol{\beta}_{1} = (0.1, 2, 0)^{\rT}$ and $\boldsymbol{\beta}_{2} = (0.1, 3, 0)^{\rT}$.
    
    \item The second setting is adapted from \citet{lin2019regression} to introduce interdependence among the baseline confounders. A vector of five continuous covariates is generated from a multivariate normal distribution, $\bX = (X_1, X_2, X_3, X_4, X_5)^{\rT} \sim N(\mathbf{0}, \bS)$, governed by a  correlation structure where $\bS_{ij} = 1$ if $i = j$ and $\bS_{ij} = 0.25$ otherwise. The corresponding treatment assignment and infection probabilities are governed by:
    $
    p_{\bx} = \text{logit}^{-1}(3 -\boldsymbol{\beta}_{1}^{\rT}\bx)
    $ and $p_{\bx, a} = \text{logit}^{-1}(3 -\boldsymbol{\beta}_{2}^{\rT}\bx +4a)
    $,
    where $\boldsymbol{\beta}_{1} = (-2, -2, -2, 2, 2)^{\rT}$ and $\boldsymbol{\beta}_{2} = (-2, -2, 1, 2, 0)^{\rT}$.
\end{enumerate} In both settings, the baseline parameters are chosen such that the true average treatment effect, $\tau = E[p_{\bX, 1}] - E[p_{\bX, 0}]$, is negative, indicating that the active intervention successfully decreases the probability of infection. Because these non-linear transformations lack a closed analytical form, the true average treatment effect is computed numerically via a Monte Carlo integration over $25,000$ independent draws from the covariate distributions.

To construct the group-level data, we randomly partition a complete sample of $M_t$ individuals generated via the mechanisms above into $M$ pools $\{G_i\}_{i=1}^M$ of a uniform preset size $N$. Imperfect pooled test responses $\widetilde{Y}^{*}_i$ are subsequently simulated from each latent, true pooled group status $Y_i^\ast = \max_{j \in G_i} Y_{ij}$ by fixing both the assay sensitivity ($S_e$) and specificity ($S_p$) to $0.99$, following standard diagnostic misclassification rates for benchmarks \citep{lin2019regression}.

\subsection{Results} \label{simulation_res}

Our numerical analyses were conducted using a total sample size of $M_t = 25,000$ individuals. To systematically evaluate the performance of our framework across varying pooling configurations, separate simulation runs were performed by varying the uniform pool sizes across $N = 2, 3, 5, 10, 15, 20, 25$. For each configuration, we executed $S = 1,000$ independent replications. The point estimates for the ATE parameters were computed within each replication by numerically maximizing the weighted pseudo-log-likelihood function \eqref{eq:pseudo_ll}.

The empirical averages of the ATE estimates over the $S = 1,000$ simulation runs are summarized in Tables~\ref{table:m1} and \ref{table:m2} for Models 1 and 2, respectively. Within each table, we contrast the performance of the proposed estimator across three distinct data regimes: individual-level outcomes (columns labeled $Y$), latent true group outcomes (columns labeled $Y^\ast$), and observed group test outcomes (columns labeled $\widetilde{Y}^{*}$). To isolate the impact of confounding adjustment, each regime evaluates three distinct weight structures: unweighted baseline models, models incorporating true propensity weights (columns labeled $W$), and models utilizing estimated plug-in weights (columns labeled $\widehat{W}$). The root mean squared errors across the $S = 1,000$ replications are reported in parentheses alongside their point estimate averages, and their trajectories across pool sizes are visualized in Figures~\ref{fig:m1} and \ref{fig:m2}.

Tables~\ref{table:m1} and \ref{table:m2} reveal a substantial discrepancy between the unadjusted, unweighted estimators and our proposed inverse probability weighted pseudo-likelihood estimators ($W$ and $\widehat{W}$). For Model 1, the true analytical ATE is $\tau \approx -0.1228$. The unweighted estimators across both the individual-level and group-level data regimes fail completely, yielding estimates in the incorrect direction. Specifically, they erroneously predict that receiving the treatment increases the absolute probability of infection by approximately $7\%$, whereas the true causal effect is an absolute reduction of roughly $13\%$. 

We observe a similar pattern in Model 2, where the true analytical ATE is $\tau \approx -0.1750$. Here, the unweighted baseline estimators predict an absolute infection probability reduction of only $6\%$. While this estimate captures the correct sign, it significantly underestimates the true magnitude of the treatment's protective effect. Conversely, across both models, the proposed IPW pseudo-likelihood estimators yield point estimates that are closely aligned with the true ATE, particularly under low-to-moderate pool sizes. These findings provide strong empirical validation for our framework, confirming that the pool-level IPW architecture successfully purges the data of selection and confounding bias. Additionally, Figures~\ref{fig:m1} and \ref{fig:m2} illustrate the relative empirical behavior of our estimator when utilizing true versus estimated weights within this specific simulation framework, tracking their respective finite-sample variance properties.

Finally, a distinct trend emerges across both Tables~\ref{table:m1} and \ref{table:m2} as the uniform pool size expands: the weighted pseudo-likelihood estimates exhibit a gradual increase in finite-sample bias as $N$ grows large. This behavior is accompanied by a reliable, monotonic increase in the empirical standard errors, as visualized in Figures~\ref{fig:m1} and \ref{fig:m2}. This simultaneous expansion of finite-sample bias and variance under large pool configurations aligns perfectly with established principles in the group testing literature. In standard group screening applications, large pool sizes inherently compound diagnostic misclassification and dilute subject-specific information. This observed behavior directly highlights the importance of the extensive literature dedicated to identifying optimal group sizes that balance the fundamental trade-off between logistical cost savings and statistical precision \citep[see, e.g.,][]{kim2007comparison, bilder2010optimal, tebbs2013nonparametric}.

\begin{table}
\centering
\caption{Model 1 Mean Estimated ATEs from $1000$ Simulations} 
\label{table:m1}
\begin{tabular}{ c  c c c  }
\hline
\hline
$-$ & $Y$ &  $Y$, $W$ &  $Y$, $\widehat{W}$ \\
\hline
$ - $ & $ 0.0672 (  0.0361 ) $ & $ -0.1223 (  0.0002 ) $ & $ -0.1224 (  0.0002 ) $  \\  
\hline
$N$  &$Y^\ast$ &$Y^\ast$, $W$  &$Y^\ast$, $\widehat{W}$\\
\hline
$ 2 $ & $ 0.0672 (  0.0361 ) $ & $ -0.1208 (  0.0006 ) $ & $ -0.1209 (  0.0006 ) $  \\ 
$ 3 $ & $ 0.0672 (  0.0361 ) $ & $ -0.1208 (  0.0018 ) $ & $ -0.1210 (  0.0018 ) $  \\ 
$ 5 $ & $ 0.0672 (  0.0361 ) $ & $ -0.1146 (  0.0028 ) $ & $ -0.1144 (  0.0027 ) $  \\ 
$ 10 $ & $ 0.0672 (  0.0361 ) $ & $ -0.0935 (  0.0058 ) $ & $ -0.0937 (  0.0058 ) $  \\ 
$ 15 $ & $ 0.0671 (  0.0361 ) $ & $ -0.0713 (  0.0067 ) $ & $ -0.0714 (  0.0067 ) $  \\ 
$ 20 $ & $ 0.0671 (  0.0361 ) $ & $ -0.0585 (  0.0083 ) $ & $ -0.0586 (  0.0083 ) $  \\ 
$ 25 $ & $ 0.0670 (  0.0360 ) $ & $ -0.0469 (  0.0095 ) $ & $ -0.0468 (  0.0095 ) $  \\  
\hline
$N$  &$\widetilde{Y}^{*}$ &$\widetilde{Y}^{*}$, $W$  &$\widetilde{Y}^{*}$, $\widehat{W}$\\
\hline
$ 2 $ & $ 0.0669 (  0.0360 ) $ & $ -0.1213 (  0.0008 ) $ & $ -0.1214 (  0.0008 ) $  \\ 
$ 3 $ & $ 0.0671 (  0.0361 ) $ & $ -0.1211 (  0.0031 ) $ & $ -0.1213 (  0.0030 ) $  \\ 
$ 5 $ & $ 0.0673 (  0.0362 ) $ & $ -0.1162 (  0.0052 ) $ & $ -0.1158 (  0.0051 ) $  \\ 
$ 10 $ & $ 0.0673 (  0.0362 ) $ & $ -0.0908 (  0.0096 ) $ & $ -0.0909 (  0.0096 ) $  \\ 
$ 15 $ & $ 0.0670 (  0.0361 ) $ & $ -0.0685 (  0.0112 ) $ & $ -0.0686 (  0.0113 ) $  \\ 
$ 20 $ & $ 0.0670 (  0.0361 ) $ & $ -0.0533 (  0.0146 ) $ & $ -0.0532 (  0.0147 ) $  \\ 
$ 25 $ & $ 0.0670 (  0.0361 ) $ & $ -0.0416 (  0.0162 ) $ & $ -0.0414 (  0.0162 ) $  \\ 
\hline
\end{tabular} 
\end{table}
	
\begin{table}
\centering
\caption{Model 2 Mean Estimated ATEs from $1000$ Simulations} 
\label{table:m2}
\begin{tabular}{ c c c c  }
\hline
\hline
$-$ & $Y$ &  $Y$, $W$ &  $Y$, $\widehat{W}$ \\
\hline
$ - $ & $ -0.0567 (  0.0140 ) $ & $ -0.1746 (  0.0007 ) $ & $ -0.1748 (  0.0006 ) $  \\  
\hline
$N$  &$Y^\ast$ &$Y^\ast$, $W$  &$Y^\ast$, $\widehat{W}$\\
\hline
$ 2 $ & $ -0.0567 (  0.0140 ) $ & $ -0.1667 (  0.0031 ) $ & $ -0.1672 (  0.0032 ) $  \\ 
$ 3 $ & $ -0.0567 (  0.0140 ) $ & $ -0.1602 (  0.0056 ) $ & $ -0.1605 (  0.0057 ) $  \\ 
$ 5 $ & $ -0.0567 (  0.0140 ) $ & $ -0.1552 (  0.0089 ) $ & $ -0.1550 (  0.0088 ) $  \\ 
$ 10 $ & $ -0.0567 (  0.0140 ) $ & $ -0.1286 (  0.0075 ) $ & $ -0.1286 (  0.0075 ) $  \\ 
$ 15 $ & $ -0.0568 (  0.0140 ) $ & $ -0.1237 (  0.0083 ) $ & $ -0.1235 (  0.0084 ) $  \\ 
$ 20 $ & $ -0.0568 (  0.0140 ) $ & $ -0.1157 (  0.0090 ) $ & $ -0.1159 (  0.0089 ) $  \\ 
$ 25 $ & $ -0.0567 (  0.0140 ) $ & $ -0.1086 (  0.0097 ) $ & $ -0.1085 (  0.0097 ) $  \\ 
\hline
$N$  &$\widetilde{Y}^{*}$ &$\widetilde{Y}^{*}$, $W$  &$\widetilde{Y}^{*}$, $\widehat{W}$\\
\hline
$ 2 $ & $ -0.0569 (  0.0139 ) $ & $ -0.1648 (  0.0039 ) $ & $ -0.1653 (  0.0041 ) $  \\ 
$ 3 $ & $ -0.0566 (  0.0140 ) $ & $ -0.1604 (  0.0096 ) $ & $ -0.1605 (  0.0097 ) $  \\ 
$ 5 $ & $ -0.0564 (  0.0141 ) $ & $ -0.1510 (  0.0142 ) $ & $ -0.1509 (  0.0141 ) $  \\ 
$ 10 $ & $ -0.0567 (  0.0140 ) $ & $ -0.1242 (  0.0156 ) $ & $ -0.1238 (  0.0157 ) $  \\ 
$ 15 $ & $ -0.0570 (  0.0140 ) $ & $ -0.1083 (  0.0168 ) $ & $ -0.1081 (  0.0170 ) $  \\ 
$ 20 $ & $ -0.0565 (  0.0141 ) $ & $ -0.0979 (  0.0177 ) $ & $ -0.0981 (  0.0176 ) $  \\ 
$ 25 $ & $ -0.0564 (  0.0141 ) $ & $ -0.0925 (  0.0188 ) $ & $ -0.0925 (  0.0189 ) $  \\ 
\hline
\end{tabular} 
\end{table}

\begin{figure}
\centering

\begin{minipage}{0.4\textwidth}
    \centering
    {\small \textbf{A} Using true group outcomes $Y^{\ast}$}
    \includegraphics[width=\linewidth]{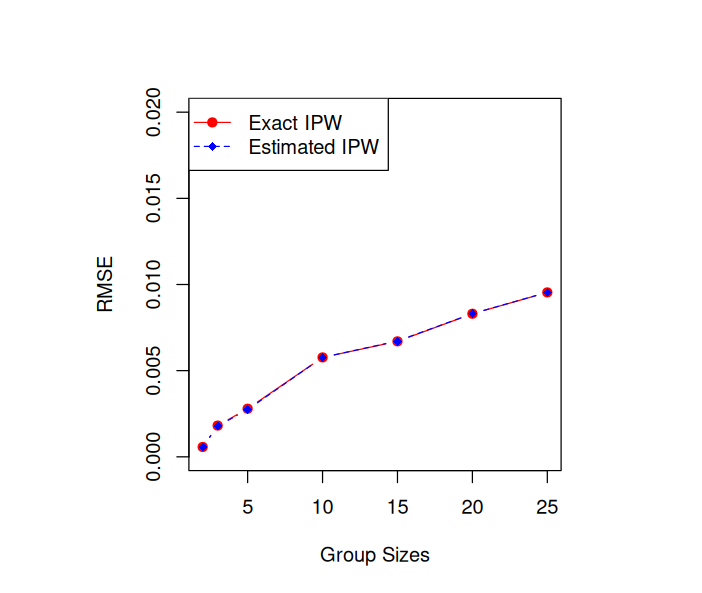}
    \label{fig:m1a}
\end{minipage}
\hfill
\begin{minipage}{0.4\textwidth}
    \centering
    {\small \textbf{B} Using group test outcomes $\widetilde{Y}^{*}$}
    \includegraphics[width=\linewidth]{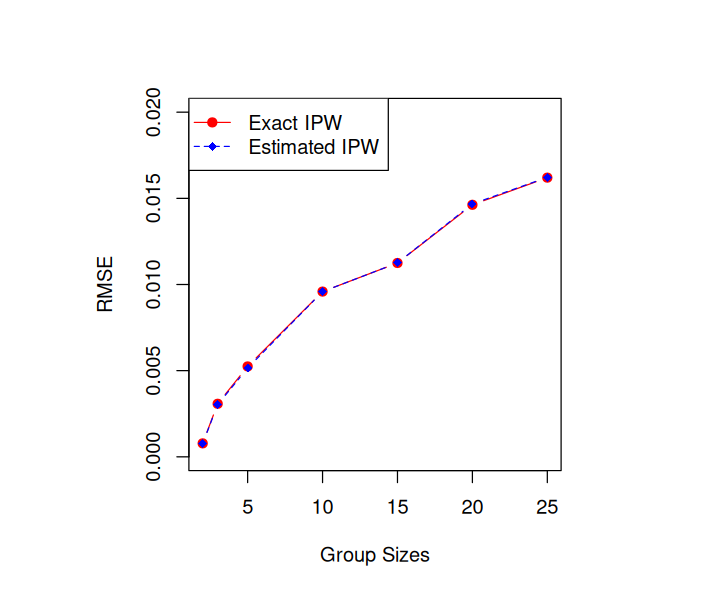}
    \label{fig:m1b}
\end{minipage}

\caption{RMSE of estimated ATEs for Model 1 based on $S=1000$ simulations}
\label{fig:m1}

\end{figure}

\begin{figure}
\centering

\begin{minipage}{0.4\textwidth}
    \centering
    {\small \textbf{A} Using true group outcomes $Y^{\ast}$}
    \includegraphics[width=\linewidth]{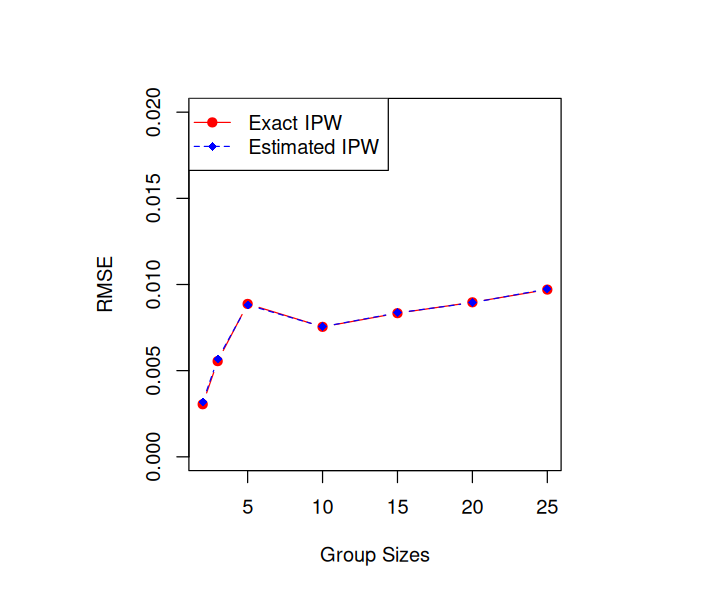}
    \label{fig:m2a}
\end{minipage}
\hfill
\begin{minipage}{0.4\textwidth}
    \centering
    {\small \textbf{B} Using group test outcomes $\widetilde{Y}^{*}$}
    \includegraphics[width=\linewidth]{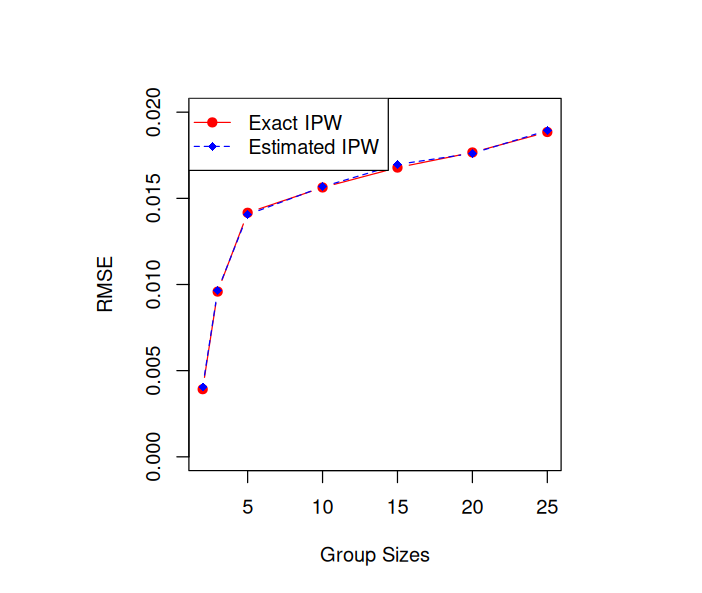}
    \label{fig:m2b}
\end{minipage}

\caption{RMSE of estimated ATEs for Model 2 based on $S=1000$ simulations}
\label{fig:m2}

\end{figure}

\section{Case Study} \label{casestudy}

\subsection{The 2018–19 US Flu VE Network Dataset}

The US Flu Vaccine Effectiveness (VE) Network is a multi-site, prospective surveillance program overseen by the Centers for Disease Control and Prevention (CDC). Its primary objective is to provide annual, real-world estimates of influenza vaccine effectiveness. Unlike controlled clinical trials, this network evaluates the vaccine’s performance in a highly heterogeneous, observational population. For this application, we utilize the database from the 2018–2019 influenza season, which comprises a total sample size of $M_t = 10,443$ enrolled participants \citep{USFluVE2018}.

Individual-level baseline information in this surveillance database includes age in years at enrollment, sex derived from electronic medical records (EMR), self-rated health status recorded on a 1–5 ordinal scale, and four distinct ethnicity-derived categories. Individual vaccination status  was determined hierarchically by prioritizing verified EMR documentation followed by validated patient self-reports. Individual-level true infection status  was determined through site interpretation of a standard reverse transcription polymerase chain reaction (RT-PCR) assay for influenza. 

The 2018–19 US Flu VE Network dataset provides an appropriate environment for evaluating causal group testing via Inverse Probability Weighting (IPW). Previous studies using data from the same network shows that this specific season is notable because of the significant heterogeneity in vaccine response across demographic groups; for instance, while overall vaccine effectiveness was moderate, protection levels varied sharply based on patient age and health status \cite{Chung2020_2018Data}. This stands in contrast to studies from more uniform seasons, such as the 2014–15 season, where vaccine mismatch was found to be fairly equally widespread across all demographic strata \citep{Flannery2015}. Thus, by applying our method to this heterogeneous season, we can determine whether measured covariates—such as pre-existing conditions and age—function as true confounders or if the observed vaccine outcomes are robust to selection bias. Causal methodologies for this manner of investigation into vaccine effectiveness already exist \citep{Sullivan2016}, but we note that such methodologies are, as yet, only applicable to individual-level observations.

\subsection{Group Test Setting}

Because the 2018–19 US Flu VE Network dataset provides individual-level outcome information, we implemented a synthetic group testing overlay on this real-world surveillance database to evaluate our framework. For each fixed pool configuration, the $M_t = 10,443$ individuals were randomly partitioned into non-overlapping groups $\{G_i\}_{i=1}^M$ of a uniform preset size $N$. Group-level infection tests ($\widetilde{Y}_i^{*}$) were subsequently simulated from the latent, true pooled group infection statuses $Y_i^\ast = \max_{j \in G_i} Y_{ij}$. To maintain consistency with our previous numerical studies and account for realistic diagnostic misclassification, the assay specificity ($S_p$) and sensitivity ($S_e$) were both fixed at 0.99. 

We evaluated a variety of uniform pool sizes ($N = 2, 3, 5, 10, 15, 20, 25$) to systematically investigate the impact of information dilution on the accuracy of our causal ATE estimates. To quantify the sampling variability of our estimators under each pooling structure and construct valid empirical variance metrics, we implemented a bootstrap resampling procedure executed across 1000 replications.

\subsection{Results} \label{casestudy_res}

The empirical ATE point estimates and their corresponding uncertainty metrics for the influenza surveillance data are compiled in Table~\ref{table:casestudy}. Mirroring the structural layout of our numerical simulations, we present the causal estimates across three data regimes: individual-level outcomes (columns labeled $Y$), latent true group statuses (columns labeled $Y^\ast$), and observed, misclassified group test outcomes (columns labeled $\widetilde{Y}^{*}$). Within each data regime, we evaluate the discrepancies between the unweighted baseline models and the models utilizing estimated plug-in weights (columns labeled $\widehat{W}$). Because these data are drawn from an observational cohort, the true population propensity scores are  unknown and must be approximated using our estimated weights, contrasting with the controlled parametric setups of Section~\ref{simulation_res}. The bootstrap standard errors of the ATE estimates, calculated across $B = 1000$ non-parametric bootstrap replications, are reported in parentheses alongside each point estimate and are visually tracked across pool sizes in Figure~\ref{fig:casestudy}.

Regarding the impact of information dilution, the estimation results exhibit the same broad structural trends observed in our controlled simulation experiments. As illustrated in Figure~\ref{fig:casestudy}, the bootstrap standard errors expand reliably and monotonically as the uniform pool size $N$ increases. However, in contrast to the divergence observed in Section~\ref{simulation_res}, the unweighted baseline estimates and our proposed weighted pseudo-likelihood estimates ($\widehat{W}$) are remarkably close across all evaluated pool sizes. Both specifications consistently indicate that the 2018–19 influenza vaccine reduced the absolute marginal infection rate by approximately $5\%$ to $6\%$. 

In the absence of an unconfounded experimental benchmark, this close empirical concordance carries an important epidemiological implication: it indicates that after conditioning on the measured baseline attributes—age, sex, and self-rated health status—the empirical vaccine effectiveness estimates within the 2018–19 network are largely robust to overt selection bias. The observed protective effect can be more confidently attributed to biological vaccine efficacy rather than being artificially driven by health-seeking behaviors or demographic imbalances among vaccinated individuals.

\begin{table}
\centering
\caption{Estimated ATEs for the 2018-19 US Flu Vaccination Data with Bootstrap SEs from $B=1000$ resamples} 
\label{table:casestudy}
\begin{tabular}{ c c c  }
\hline
\hline
$-$ & $Y$ &  $Y$, $\widehat{W}$ \\
\hline
$ - $ & $ -0.0642 (  0.0086 ) $ & $ -0.0507 (  0.0100 ) $  \\ 
\hline
$N$  &$Y^\ast$   &$Y^\ast$, $\widehat{W}$\\
\hline
$ 2 $ & $ -0.0642 (  0.0086 ) $ & $ -0.0541 (  0.0089 ) $  \\ 
$ 3 $ & $ -0.0642 (  0.0090 ) $ & $ -0.0502 (  0.0094 ) $  \\ 
$ 5 $ & $ -0.0642 (  0.0089 ) $ & $ -0.0449 (  0.0098 ) $  \\ 
$ 10 $ & $ -0.0643 (  0.0088 ) $ & $ -0.0474 (  0.0104 ) $  \\ 
$ 15 $ & $ -0.0643 (  0.0087 ) $ & $ -0.0440 (  0.0118 ) $  \\ 
$ 20 $ & $ -0.0644 (  0.0086 ) $ & $ -0.0477 (  0.0134 ) $  \\ 
$ 25 $ & $ -0.0642 (  0.0088 ) $ & $ -0.0541 (  0.0154 ) $  \\ 
\hline
$N$  &$\widetilde{Y}^{*}$  &$\widetilde{Y}^{*}$, $\widehat{W}$\\
\hline
$ 2 $ & $ -0.0612 (  0.0088 ) $ & $ -0.0504 (  0.0091 ) $  \\ 
$ 3 $ & $ -0.0640 (  0.0092 ) $ & $ -0.0506 (  0.0097 ) $  \\ 
$ 5 $ & $ -0.0644 (  0.0092 ) $ & $ -0.0464 (  0.0101 ) $  \\ 
$ 10 $ & $ -0.0642 (  0.0089 ) $ & $ -0.0469 (  0.0107 ) $  \\ 
$ 15 $ & $ -0.0625 (  0.0091 ) $ & $ -0.0444 (  0.0122 ) $  \\ 
$ 20 $ & $ -0.0653 (  0.0089 ) $ & $ -0.0460 (  0.0138 ) $  \\ 
$ 25 $ & $ -0.0644 (  0.0089 ) $ & $ -0.0524 (  0.0157 ) $  \\  
\hline
\end{tabular} 
\end{table}

\begin{figure}
\centering

\begin{minipage}{0.4\textwidth}
    \centering
    {\small \textbf{A} Using true group outcomes $Y^{\ast}$}
    \includegraphics[width=\linewidth]{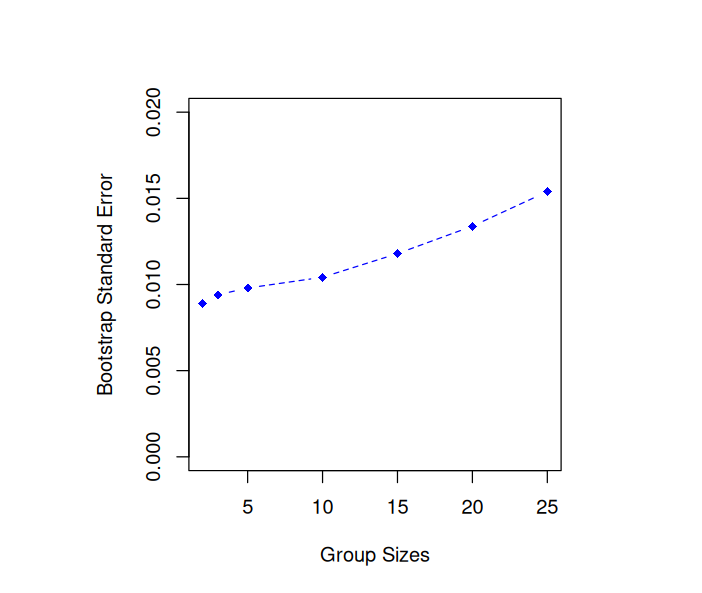}
    \label{fig:casea}
\end{minipage}
\hfill
\begin{minipage}{0.4\textwidth}
    \centering
     {\small \textbf{B} Using group test outcomes $\widetilde{Y}^{*}$}
    \includegraphics[width=\linewidth]{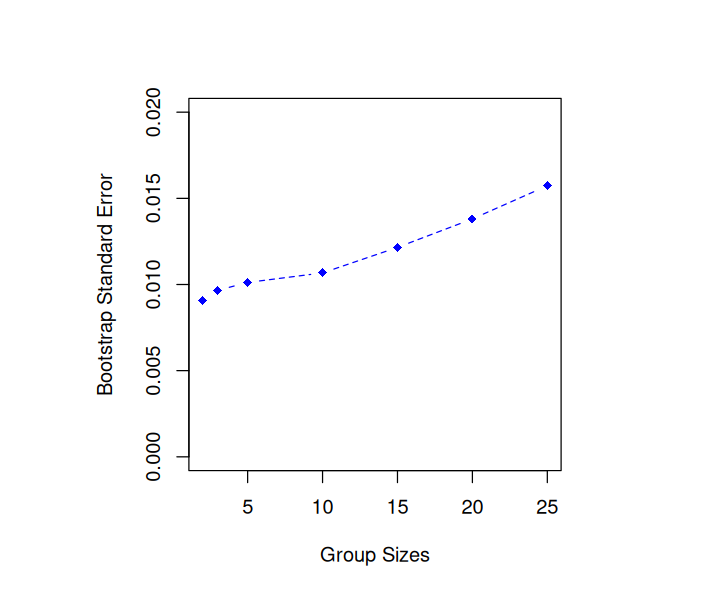}
    \label{fig:caseb}
\end{minipage}

\caption{2018-19 US Flu Data RMSE of Estimated ATEs based on $B=1000$ Bootstraps}
\label{fig:casestudy}

\end{figure}


\section{Discussion} \label{discussion}

In this work, we introduced a novel causal inference framework for group testing data by integrating inverse probability weighting (IPW) into a pooled pseudo-likelihood formulation. Our numerical studies demonstrate that when baseline confounding is present, traditional group testing estimators fail to capture the true average treatment effect (ATE). In severe cases of selection bias, unadjusted models can even yield estimates in the incorrect direction. By incorporating pool-level IPW adjustments, our proposed framework successfully recovers estimates that align with the true analytical ATE, under low-to-moderate pool sizes.

The simulation results revealed a clear trade-off: while expanding the uniform pool size $N$ inherently reduces the testing burden, it introduces finite-sample bias and increases the empirical standard errors of the ATE estimates. This behavior suggests that in causal applications of group testing, investigators must carefully balance financial constraints against their statistical tolerance for finite-sample bias in the causal effect estimate.

Our application to the 2018–19 US Flu VE Network surveillance database demonstrated an alternative utility for our method. In the case study, the differences between the weighted and unweighted ATE estimates were insignificant. This empirical concordance implies that for this specific influenza season, selection bias driven by the measured baseline attributes—age, sex, and self-rated health status—was minimal, or that the vaccine's protective effect was relatively uniform across the sampled population. Thus, our proposed framework can function not only as an estimation tool, but also as a valuable diagnostic asset for investigating the presence of confounding mechanisms within routine public health screening registries.

Future research can extend this methodology in several promising directions. First, building directly upon our findings regarding the constraints of large pool sizes, a next step is to derive objective functions to determine optimal group sizes that balance testing budgets against statistical precision. Second, a natural progression within the causal domain is to develop a ``doubly robust'' extension by marrying our pool-level IPW pseudo-score with an individual or group-level outcome regression model. Such an architecture would provide valid causal inferences if either the propensity score model or the outcome model is correctly specified, offering protection against model misspecification. Finally, more complex causal topologies could be explored to account for network interference or mediation pathways, where an individual's exposure or treatment within a given pool impacts the outcomes of their peers.

Ultimately, our framework represents an initial foray into bridging group testing and causal inference—two vast, independent fields that are both of increasing relevance to modern epidemiological practice. We hope this work serves as a foundation that inspires further methodology in this fruitful intersection.

\bibliographystyle{unsrtnat} \bibliography{biom}

\clearpage

\begin{appendices}

\section{Likelihoods of Randomized Controlled Trials} \label{rct_likelihood}
		
		Here we derive the likelihood for group testing in the context of RCTs. We seek to express the likelihood in terms of quantities depending only on the observed group-level quantities $\bD^{(0)} = \{(\widetilde{Z}_i, \bA_i), i = 1, \ldots, M\}$ and the parameter $\bp = \{E[Y^{(a)}]\}_{a\in \mc{A}}$.

        Technically for the derivations, we need to assume that pool-level sensitivity and specificity are intrinsic properties of the assay that do not shift across different treatment states or results of the retesting phase i.e. $\widetilde{Y}^{*} \independent \bA, \widetilde{\bY} \ | \ Y^{*}$. Similarly, we assume $\widetilde{Y}_j \independent A_j \ | \ Y_j$ for $j=1 , \ldots, N$ for the individual-level sensitivity and specificity during the retesting phase. 
		
		We first express the conditional probability of the getting a negative group in terms of the observed treatments and the parameter as
		$$
		\eta_{0}(\ba;\bp) := P[Y^{*}=0|\bA=\ba; \bp]=\prod_{j=1}^{N}(1-P[Y_{j}=1|A_{j}=a_{j}; \bp])=
		\prod_{j=1}^{N}(1-p_{a_{j}}). \label{eq:eta0}
		$$
		This allows expression of the conditional probability of getting a positive group through those same quantities
		$$
		\eta_{1}(\ba;\bp) := P[Y^{*}=1|\bA=\ba ; \bp] = 1 - \eta_0(\ba, \bp).
		$$ Consequently, we can write the probability of getting a negative group test as
		\begin{align*}
			\widetilde{\eta}_0(\ba;\bp) := &P[\widetilde{Y}^{*}=0|\bA=\ba; \bp] \\
			=&P[Y^{*}=0|\bA=\ba; \bp]S_{p} +  P[Y^{*}=1|\bA=\ba; \bp](1-S_{e})    \\
			= &\eta_0(\ba; \bp)S_p + \eta_1(\ba; \bp)(1-S_e) \label{eq:etatilde0}
		\end{align*}
		
		Next, we define the individual retest specificity and sensitivity as
		\begin{align*}
			&S_{p}(\tilde{y}_{j}):= 1\{\tilde{y}_{j} = 0\}S_p + 1\{\tilde{y}_{j} = 1\}(1 -S_p), \\
			&S_{e}(\tilde{y}_{j}):=1\{\tilde{y}_{j} = 0\}(1 - S_e) + 1\{\tilde{y}_{j} = 1\}S_e
		\end{align*} and consider the following probabilities for the retesting phase
		\begin{align*}
			&\h2s \phi_0(\widetilde{\by}, \ba; \bp) := P[ \widetilde{\bY}=\widetilde{\by}, Y^\ast = 0 |  \bA=\ba; \bp] \\
			= &P\Big[Y_{1} = 0,\cdots, Y_{N} = 0, \widetilde{\bY}=\widetilde{\by} \big|  \bA=\ba, \bp\Big] = \prod_{j = 1}^{N}P\Big[Y_{j} = 0, \widetilde{Y}_{j}=\tilde{y}_j \big| A_{j}=a_{j}; \bp\Big] \\
			= &\prod_{j = 1}^{N}P\Big[ \widetilde{Y}_{j}=\tilde{y}_j |
			Y_{j} = 0, A_{j}=a_{j}; \bp\Big] 
			P\Big[ Y_{j} = 0 | A_{j}=a_{j}; \bp\Big]\\
			= &\prod_{j = 1}^{N}P\Big[\widetilde{Y}_{j} =\tilde{y}_j| Y_{j} = 0\Big]P\Big[Y_{j} = 0 \big| A_{j}=a_{j}; \bp\Big] =  \prod_{j = 1}^{N}S_{p}(\tilde{y}_{j}) (1-p_{a_{j}}), \text{ and } \\
			&\h2s \Phi(\widetilde{\by}, \ba;\bp):= P[\widetilde{\bY}=\widetilde{\by} |  \bA=\ba; \bp] 
			= \prod_{j = 1}^{N}P\Big[\widetilde{Y}_{j} = \tilde{y}_{j}\big| A_{j}=a_{j}; \bp\Big] \\
			= &\prod_{j = 1}^{N}\Bigg\{\sum_{y = 0,1}P\Big[\widetilde{Y}_{j}=\tilde{y}_{j} \big| Y_{j} = y\Big]P\Big[Y_{j} = y \big| A_{j}=a_{j}; \bp\Big] \Bigg\}\\
			= &\prod_{j = 1}^{N}\Bigg\{ \Big(1\{\tilde{y}_{j} = 0\}S_p + 1\{\tilde{y}_{j} = 1\}(1 -S_p)\Big) (1-p_{a_{j}}) \\
			&\h2s\h2s + \Big(1\{\tilde{y}_{j} = 0\}(1 - S_e) + 1\{\tilde{y}_{j} = 1\}S_e\Big) p_{a_{j}}\Bigg\}\\
			= & \prod_{j = 1}^{N}\Big( S_{p} (\tilde{y}_{j})(1-p_{a_{j}}) + S_{e}(\tilde{y}_{j}) p_{a_{j}}\Big).
		\end{align*} The above also gives
		$$
		\phi_1(\widetilde{\by},\ba;\bp):=P[\widetilde{\bY}=\widetilde{\by}, Y^\ast = 1 |  \bA=\ba;\bp] = \Phi(\widetilde{\by},\ba;\bp) - \phi_0(\widetilde{\by},\ba; \bp)$$ which allows us to express the probability of getting a positive group test with retesting, in terms of quantities that depend on the observed data and the parameter:
		\begin{align*}
			\widetilde{\phi}_1(\widetilde{\by},\ba;\bp) := &P[\widetilde{Y}^{*}=1, \widetilde{\bY} = \widetilde{\by}|\bA = \ba; \bp] = \sum_{y = 0, 1} P[\widetilde{\bY}, \widetilde{Y}^{*}=1, Y^\ast = y|  \bA=\ba; \bp] \\
			=&\sum_{y = 0, 1} P[Y^\ast = y, \widetilde{\bY} =\widetilde{\by} |  \bA_i=\ba_i; \bp] P[\widetilde{Y}^{*}=1 | Y^\ast = y, \widetilde{\bY} = \widetilde{\by}, \bA=\ba; \bp] \\
			=&\phi_0(\widetilde{\by}, \ba; \bp) (1 - S_p) + \phi_1(\widetilde{\by}, \ba; \bp)S_e 
		\end{align*}
		
		Finally, we can express the likelihood in terms of quantities depending only on the observed data $\bD^{(0)} = \{(\widetilde{Z}_i, \bA_i), i = 1, \ldots, M\}$ and the parameter $\bp$:
		\begin{align*}
			&L(\bp|\bD^{(0)})= \prod_{i=1}^{M}\Big[P(\widetilde{Y}^{*}_i=0|\bA_i; \bp)
			\Big]^{1-\widetilde{Y}^{*}_{i}} \Big[P(\widetilde{Y}_i=1, \widetilde{\bY}_i|\bA_i, \bp)
			\Big]^{\widetilde{Y}^{*}_{i}} \\
			=&\prod_{i=1}^{M}\Big[\widetilde{\eta}_{0}(\bA_i;\bp)\Big]^{1-\widetilde{Y}^{*}_{i}}\Big[\widetilde{\phi}_1(\widetilde{\bY}_i,\bA_i;\bp)\Big]^{\widetilde{Y}^{*}_i}.
		\end{align*}

		\section{Potential Outcome Densities in Observational Studies} \label{os_derive}
		
		In the case of Observational Studies, we need to consider the baseline covariates. Thus, the assumptions for the pool-level and individual-level diagnostic accuracies now take the forms $\widetilde{Y}^{*} \independent \bA, \mb{X}, \widetilde{\bY} \ | \ Y^{*}$ and $\widetilde{Y}_j \independent A_j, \bX_j \ | \ Y_j$ for $j=1 , \ldots, N$.
        
        Now we demonstrate that the quantities from  \ref{rct_likelihood} are related to the densities of the potential outcomes instead of the conditional densities of outcomes as follows:
		\begin{align*}
			&\h2s E_{\mb{X}}\Big[P(Y^\ast = 0 | \bA = \ba, \mb{X}; \bp)\Big] = E_{\mb{X}}\Big[\prod_{j = 1}^{N}P(Y_{j} = 0 | A_{j} = a_{j}, \bX_{j}; \bp)\Big] \\
			&= \prod_{j = 1}^{N}E_{\bX_{j}}\Big[P(Y_{j} = 0 | A_{j} = a_{j}, \bX_{j}; \bp)\Big] = \prod_{j = 1}^{N} \pi_{Y^{(a_{j})}}(0; \bp) = \prod_{j = 1}^{N} (1 - p_{a_{j}}) = \eta_0(\ba, \bp)
		\end{align*} 
		and thus $E_{\mb{X}}\Big[P(\widetilde{Y}^{*}= 0 | \bA =  \ba, \mb{X}; \bp)\Big] = \widetilde{\eta}_{0}(\ba, \bp)$. 
		
		A similar derivation now shows that $E_{\mb{X}}\Big[P(\widetilde{Y}^{*} = 1, \widetilde{\bY} =\widetilde{\by} | \bA =  \ba, \mb{X}; \bp)\Big] = \widetilde{\phi}_1(\widetilde{\by}, \ba; \bp)$. Therefore, 
		\begin{align*}
			&\pi_{\widetilde{Z}^{(\ba)}}(\tilde{z}; \bp) = E_{\mb{X}}[P(\widetilde{Z}=\tilde{z} | \bA = \ba, \mb{X}; \bp)] \\
			= &E_{\mb{X}}\Big[1\{\tilde{y}^{*} = 0\}P(\widetilde{Y}^{*} = 0 | \bA = \ba, \mb{X}; \bp) + 1\{\tilde{y}^{*} = 1\}P(\widetilde{Y}^{*}= 1, \widetilde{\bY} =\widetilde{\by}  | \bA = \ba, \mb{X}; \bp)\Big]\\
			= &1\{\tilde{y}^{*} = 0\}E_{\mb{X}}\Big[P(\widetilde{Y}^{*} = 0 | \bA = \ba, \mb{X}; \bp)\Big] + 1\{\tilde{y}^{*} = 1\}E_{\mb{X}}\Big[P(\widetilde{Y}^{*} = 1, \widetilde{\bY}=\widetilde{\by} | \bA = \ba, \mb{X}; \bp)\Big] \\
			= &\Big[\widetilde{\eta}_{0}(\ba;\bp)\Big]^{1-\tilde{y}^{*}}\Big[\widetilde{\phi}_1(\widetilde{\by}, \ba; \bp)\Big]^{\tilde{y}^{*}}
		\end{align*}
		
		\section{Theoretical Results} \label{theory}
		
		We impose the following regularity conditions to facilitate our technical derivations:
		\begin{enumerate}[({A}1)]
			\item \label{Aex} \textit{Positivity}: $\pi_{a}^{\ast}(\bx)=P(A=a|\bX=\bx)>\nu$ for all values of $\bx$ that $f_{\bX}(\bx)>0$, where $f_{\bX}(\cdot)$ is the density of $\bX$ and 
			$0 < \nu < 1/2$ is some constant.
			\item \label{Abd} \textit{Bounds for $\bp^{*}$}: $0 < \nu \leq \min\{p_0^{*},p_1^{*}\} \leq \max\{p_0^{*},p_1^{*}\} \leq 1-\nu$.
			\item \label{Agp} \textit{Group size bound}: The group sizes are bounded $N_i<1/\nu$ for all $i = 1, \ldots, M$.
            \item \label{Apd} \textit{Full gradient span}: $\bG_{\bp^\ast, i} := \{\nabla_{\bp} \pi_{\widetilde{Z}_i^{(\ba_i)}}(\tilde{z}_i; \bp^\ast); \tilde{z}_i \in \mc{Z}_i, \ba_i \in \mc{A}^{N_i}\}$ has full span,
			for all $i = 1, \ldots, M$.

		\end{enumerate}
		
		\begin{remark}
			Condition (A\ref{Aex}) has been widely assumed in observational studies \citep[see, e.g.,][]{hernan2020causal,yan2019estimation}. To avoid the singularity issue around 0, Condition (A\ref{Abd}) assumes that $\bp^{*}$ is bounded away from 0. The assumption that $\bp^{*}$ is bounded below from $1/2$ is also reasonable, as there is no benefit of conducting group testing if $\bp^{*}$ is close to $1/2$. Condition (A\ref{Agp}) is necessary to ensure proper asymptotics of our method. It is reasonable as in practice group sizes are small relative to sample size.
			Condition (A\ref{Apd}) is needed for the positive definiteness of $-E[\dot{\bU}(\bp^\ast)]$, which ensures the monotonicity of the estimation equations $\bU(\bp)=0$ \citep{fygenson1994monotone} in a neighborhood around $\bp^\ast$. It is not a restrictive condition. Indeed, since the number of gradients in $\bG_{\bp^\ast, i}$ far exceed their dimension $2$, this condition is guaranteed for all but some exceptionally degenerate families of potential outcome distributions.

			
		\end{remark}
		
		\subsection{Theorems}
		

		\begin{proof}[Proof of Theorem \ref{thm1}] Without loss of generality, we assume that the group size is uniform $N_i = N$ for all $i = 1, \ldots, M$. This is a widely-used simplifying assumption employed for developing many foundational models of group testing \cite[see, e.g.,][]{dorfman1943detection}. It's main simplification is ensuring identical distribution of the observed samples, which allows use of the standard central limit theorem below.

        (i) We prove the result by invoking Theorem 3.2 in \cite{crowder1986consistency}.  Thus, we need to check the conditions (i) and (ii) of Theorem 3.2 in \cite{crowder1986consistency}.
		
		By Lemma \ref{lemma2}, 
		$\mc{F}_{b},b=0,1$'s and $\mc{G}_{bc},b,c=0,1$'s are Donsker, and hence
		Glivenko-Cantelli. Thus, 
		\begin{equation}
			\sup_{\bp\in [\nu, 1-\nu]^2}\max\left\{ \left\|\bU(\bp) - E[\bU(\bp)]\right\|,\left\|
			\dot{\bU}(\bp)-E\left[ \dot{\bU}(\bp) \right]\right\|\right\}\rightarrow_{p}0.  \label{DGC}
		\end{equation}
		Thus, condition (ii) of Theorem 3.2 in \cite{crowder1986consistency} is satisfied.

		Let $\partial B_{r}(\bp^\ast):=\{\bp \in [\nu, 1-\nu]^2: \|\bp-\bp^{*}\|=r \}$ be the boundary of the circle $B_{r}(\bp^\ast):=\{\bp \in [\nu, 1-\nu]^2: \|\bp-\bp^{*}\|\leq r \}$ for some $0 <r \leq R_+$, where $R_+$ is the fixed radius, independent of $M$, from Corollary \ref{corollary4}.  Given a $\bp\in \partial B_{r}$,  consider  $(\bp-\bp^{*})^{\rT}E[\bU(\bp)]$. Since $E[\bU(\bp^{*})]=0$ by Lemma \ref{lemma3}, by the vector valued mean value theorem \citep{mcleod1965mean,furi1991mean}, 
        \begin{align}
			&(\bp-\bp^{\ast})^{\rT}E[\bU(\bp)]=(\bp-\bp^{\ast})^{\rT}(E[\bU(\bp)]-E[\bU(\bp^{\ast})])\nonumber\\
			= &(\bp-\bp^{*})^{\rT}\left(\lambda_1E\left[\dot{\bU}(\bp_1)\right]
			+\lambda_2E\left[\dot{\bU}(\bp_2)\right]\right)(\bp-\bp^{*}) 
			\label{smvt}
		\end{align}
		where $\lambda_1, \lambda_2>0, \lambda_1+\lambda_2=1$, $\bp_1=t_{1}\bp+(1-t_1)\bp^{\ast}$, and
		$\bp_2=t_{2}\bp+(1-t_2)\bp^{\ast}$, for some $0<t_1,t_2<1$.
		
		Under Condition (A\ref{Apd}), Corollary \ref{corollary4} to Lemma \ref{lemma4} implies that the eigenvalues of $-E[\dot{\bU}(\bp)]$ are greater than $\Lambda$ for some $\Lambda>0$, uniformly over $\bp\in [\nu, 1-\nu]^2$.  Therefore, combining \eqref{smvt} and the fact above yield $
		(\bp-\bp^{*})^{\rT}(-E[\bU(\bp)])\geq \Lambda\|\bp-\bp^{*}\|^2= \Lambda r^2
		$. Thus, condition (ii) of Theorem 3.2 in \cite{crowder1986consistency} is satisfied as well.
		
		By Theorem  3.2 in \cite{crowder1986consistency}, there exists a sequence $\tilde{\bp}$ such that $\tilde{\bp}\rightarrow_{p} \bp^{*}$. 
		
		\noindent (ii)  By the vector valued mean value theorem \citep{mcleod1965mean,furi1991mean} again, 
		$$ 0=\bU(\tilde{\bp})=\bU(\bp^{*})+\left(\lambda_1 \dot{\bU}(\tilde{\bp}_1)
		+\lambda_2\dot{\bU}(\tilde{\bp}_2)\right)(\tilde{\bp}-\bp^{*}), 
		$$
		where $\lambda_1, \lambda_2>0, \lambda_1+\lambda_2=1$, $\tilde{\bp}_1=t_{1}\tilde{\bp}+(1-t_1)\bp^{*}$, and
		$\tilde{\bp}_2=t_{2}\tilde{\bp}+(1-t_2)\bp^{*}$, for some $0<t_1,t_2<1$. 
		Thus, $$\sqrt{M}(\tilde{\bp}-\bp^{*})=-\left(\lambda_1 \dot{\bU}(\tilde{\bp}_1)
		+\lambda_2\dot{\bU}(\tilde{\bp}_2)\right)^{-1}\sqrt{M}\bU(\bp^{*}).$$
		By central limit theorem,  
		$\sqrt{M}\bU(\bp^{*})\rightarrow_{d} N(0, \bUS)$, where
		$\bUS = M E\left[\bU(\bp^{*})\bU(\bp^{*})^{\rT}\right]$. Because  $\tilde{\bp}\rightarrow_{p} \bp^{*}$ by Part (i) above, applying the continuous mapping theorem gives $\left(\lambda_1 \dot{\bU}(\tilde{\bp}_1)
		+\lambda_2\dot{\bU}(\tilde{\bp}_2)\right)^{-1}$ $\rightarrow_{p}$ $\bV$ where $\bV = E[ \dot{\bU}(\bp^{*})]^{-1}$. Then by Slutsky's theorem,
		$$
		\sqrt{M}(\tilde{\bp}-\bp^{*})\rightarrow_{d} N\Big(0, 
		\bV \bUS \bV^{\rT}\Big).
		$$
        
        This completes the proof of Theorem \ref{thm1}.

        \end{proof}

		\begin{proof}[Proof of Theorem \ref{cor1}] Consider
		\begin{align*}
			\widehat{U}_b(\bp) - U_b(\bp) = \frac{1}{M} \sum_{i = 1}^M(\widehat{W}(\bA_i, \mb{X}_i) - W(\bA_i, \mb{X}_i)) \partial_{p_b}\ell(\bp|\widetilde{Z}_i, \bA_i) 
		\end{align*} For the first step, we seek to show that $\sup_{\bp \in [\nu,1-\nu]^2}|\widehat{U}_{b}(\bp)- U_{b}(\bp)|=o_{p}(1)$. Since $\|\widehat{\pi}_1-\pi_1^{*}\|_{\infty}=o_{p}(1)$ already implies
		$\max_{1\leq i\leq M} \|\widehat{W}_{i}- W_{i}\|=o_{p}(1)$, this can be achieved if $\partial_{p_b}\ell(\bp|\widetilde{Z}_i, \bA_i) = O_p(1)$.  As
		\begin{align*}
			\partial_{p_b}\ell(\bp|\widetilde{Z}_i, \bA_i) =  \left[ (1-\widetilde{Y}^{*}_i)\frac{\partial_{p_b}\widetilde{\eta}_0(\bA_i;\bp)}{\widetilde{\eta}_0(\bA_i;\bp)} +\widetilde{Y}^{*}_i \frac{\partial_{p_b}\widetilde{\phi}_1(\widetilde{\bY}_i, \bA_i; \bp)}{\widetilde{\phi}_1(\widetilde{\bY}_i, \bA_i; \bp)} \right]
		\end{align*} it suffices to show that the denominator terms $\widetilde{\eta}_0(\bA_i;\bp), \widetilde{\phi}_1(\widetilde{\bY}_i, \bA_i; \bp)$ are uniformly bounded away from $0$ and indeed this is guaranteed by Lemma \ref{lemma1}. Thus  
		\begin{align*}
			\sup_{\bp \in [\nu,1-\nu]^2}|\widehat{U}_{b}(\bp)- U_{b}(\bp)|=o_{p}(1) 
		\end{align*} as desired.
		
		Since both $U_{0}(\bp)$ and $\widehat{U}_{0}(\bp)$ are bounded, 
		by dominated convergence theorem \citep[see, e.g.,][]{Williams1991Prob},  $\sup_{\bp \in [\nu,1-\nu]^2}E[ | \widehat{U}_{0}(\bp)-U_0(\bp)|]\rightarrow 0$. Similarly,   $\sup_{\bp \in [\nu,1-\nu]^2}E[| \widehat{U}_{1}(\bp)-U_{1}(\bp)|]\rightarrow 0$.  Consequently, we obtain that
		\begin{equation}
			\sup_{\bp \in [\nu,1-\nu]^2}E[ \| \widehat{\bU}(\bp)-\bU(\bp)\|]\rightarrow 0. \label{cor1eq2}
		\end{equation}
		
		Thus, letting $\partial B_{r}(\bp^\ast):=\{\bp \in [\nu, 1-\nu]^2: \|\bp-\bp^{*}\|=r \}$ be the boundary of the circle $B_{r}(\bp^\ast):=\{\bp \in [\nu, 1-\nu]^2: \|\bp-\bp^{*}\|\leq r \}$ for some $0 <r \leq R_+$, where $R_+$ is the radius from Corollary \ref{corollary4}, 
		\begin{align*}
			&\h2s \inf_{\bp \in \partial B(r)}
			(\bp-\bp^{*})^{\rT}(-E[\widehat{\bU}(\bp)]) \\
			&\geq
			\inf_{\bp \in \partial B(r)}
			(\bp-\bp^{*})^{\rT}(-E[\bU(\bp)])-\sup_{\bp \in \partial B(r)}\left\|
			\bp-\bp^{*}
			\right\|\|E[\bU(\bp)]-E[\widehat{\bU}(\bp)] \| \\ & \geq \Lambda\|\bp-\bp^{*}\|^2.
		\end{align*}
		where we get the last inequality because under Condition (A\ref{Apd}), Corollary \ref{corollary4} to Lemma \ref{lemma4} implies that the eigenvalues of $-E[\dot{\bU}(\bp)]$ are greater than $\Lambda$ for some $\Lambda>0$, uniformly over $\bp\in [\nu, 1-\nu]^2$. Thus, condition (i) of Theorem 3.2 in \cite{crowder1986consistency} is satisfied.
		
		Moreover,  by \ref{thm1}, and \eqref{cor1eq2},
		\begin{align*}
			&\h2s \sup_{\bp \in \partial B(r)}\left\| 
			\widehat{\bU}(\bp)-E\left[\widehat{\bU}(\bp)\right]
			\right\|\\
			&\leq \sup_{\bp \in \partial B(r)}\left\| 
			\widehat{\bU}(\bp)-\bU(\bp)
			\right\| + \sup_{\bp \in \partial B(r)}\left\| 
			\bU(\bp)-E\left[\bU(\bp)\right]
			\right\| \\ & \h2s\h2s + \sup_{\bp \in \partial B(r)}\left\| 
			E\left[\bU(\bp)\right]-E\left[\widehat{\bU}(\bp)\right]
			\right\|\rightarrow_{p} 0.
		\end{align*}
		Thus, condition (ii) of Theorem 3.2 in \cite{crowder1986consistency} is satisfied. Therefore,  there exists a sequence $\widehat{\bp}$ such that $\widehat{\bp}\rightarrow_{p} \bp^{*}$. This completes the proof of Theorem \ref{cor1}.
		    
		\end{proof}

		\begin{proof}[Proof of Theorem \ref{thm2}] We aim to verify that our case satisfies conditions (2.1) - (2.6) of Theorem 2 of \cite{chen2003estimation}.
		Let $\bU(\bp, \bUpi)=(U_{0}(\bp, \bUpi), U_{1}(\bp, \bUpi))^{\rT}$, where
		\begin{align*}
			&U_{b}(\bp,\bUpi):=\mb{E}_{M}\Bigg[
			\prod_{k=1}^{N_i}\left(
			\sum_{a=0}^{1}1\{A_{ik}=a\}\pi_{a}^{-1}(\bX_{ik})\right) \left\{ 
			\partial_{p_b}\ell(\bp|\widetilde{Z}_i, \bA_i) \right\}\Bigg]
		\end{align*}
		Thus,  $\bU(\bp, \bUpi^{*})=\bU(\bp)$ and $\bU(\bp, \widehat{\bUpi})=\widehat{\bU}(\bp)$. 
		
		It is easy to see that conditions (2.1) and (2.4) are satisfied.  Since $\Gamma_{1}(\bp,\bUpi^{*})=E[\dot{\bU}(\bp)]$ exists for $\bp\in B_{r}(\bp^\ast)$ and is continuous and full rank at $\bp=\bp^{*}$, condition (2.2) is also satisfied.

		Next, we check condition (2.3). For all $\bp\in B_{r}(\bp^\ast)$, we define the function derivative of $E[\bU(\bp,\bUpi)]$ at $\bUpi^{*}$ in the direction of $[\bUpi-\bUpi^{*}]$ as 
		\begin{align*}
			&\h2s \Gamma_{2}(\bp, \bUpi^{*})[\bUpi-\bUpi^{*}] = \lim_{\tau\rightarrow 0}\frac{1}{\tau}E\Bigg[
			\Bigg(
			\prod_{k=1}^{N_i}\left(
			\sum_{a=0}^{1}1\{A_{ik}=a\}(\pi_{a}^{*}+\tau(\pi_{a}-\pi_{a}^{*}))^{-1}(\bX_{ik})\right) \\
			&\hspace{0.5in}  -\prod_{k=1}^{N_i}\left(
			\sum_{a=0}^{1}1\{A_{ik}=a\}\pi_{a}^{*-1}(\bX_{ik})\right)  
			\Bigg)\times
			\left\{ 
			\partial_{p_b}\ell(\bp|\widetilde{Z}_i, \bA_i) \right\}
			\Bigg]\\
			&=E\Bigg[
			h(D_i, \bUpi)
			\left\{ 
			\partial_{p_b}\ell(\bp|\widetilde{Z}_i, \bA_i)\right\}
			\Bigg], \text{ where}
			\\
			& h(D_i, \bUpi) = \sum_{k=1}^{N_{i}}\Bigg[\Bigg(\sum_{a=0}^{1}1\{A_{ik}=a\} \frac{(\pi_{a}^{*}-\pi_{a})}{\pi_{a}^{*2}}(\bX_{ik})\Bigg)
			\prod_{l\neq k}\left(
			\sum_{a=0}^{1}1\{A_{il}=a\}\pi_{a}^{*-1}(\bX_{il})\right) 
			\Bigg].
		\end{align*}
		$\Gamma_{2}(\bp, \bUpi^{*})[\bUpi-\bUpi^{*}]$ exists for all $\bUpi$ satisfying $\mc{H}_{\delta}:= \{\pi_1:\|\pi_1-\pi_1^{*}\|_{\infty} \leq\delta\}$. It is easy to see that $\|h(D_i, \bUpi)\|_{\infty}=O(\delta)$.
		Following the routine Taylor expansion, we can show that 
		$$
		\left\|E\left[\bU(\bp, \bUpi)\right]-E\left[\bU(\bp, \bUpi^{*})\right]-\Gamma_{2}(\bp, \bUpi^{*})[\bUpi-\bUpi^{*}] \right\|= O(\|\pi_1-\pi^{*}_1\|_{\infty}^2)
		$$
		and  if $\|\bp-\bp^{*}\|\leq \delta$,
		$\|\Gamma_{2}(\bp, \bUpi^{*})[\bUpi-\bUpi^{*}]-\Gamma_{2}(\bp^{*}, \bUpi^{*})[\bUpi-\bUpi^{*}]\|=
		O(\delta^2)$.
		Thus, condition (2.3) is satisfied.

		We  follow Theorem 3 in \cite{chen2003estimation} to verify condition (2.5'). 
		Let 
		\begin{align*}
			& m_{c}(D_i,\bp,\bUpi)= \left(\prod_{k=1}^{N_i}
			\sum_{a=0}^{1}1\{A_{ik}=a\}\pi_{a}^{-1}(\bX_{ik})\right) \times\left\{ \partial_{p_c}\ell(\bp|\widetilde{Z}_i, \bA_i) \right\},
			\h2s c=0,1.
		\end{align*}
		By Condition (A\ref{Agp}),
		\begin{align*}
			&\h2s    \left|m_{c}(D_i,\bp_1,\bUpi_1) -  m_{c}(D_i,\bp_2,\bUpi_2)\right|\\
			&\leq \left|m_{c}(D_i,\bp_1,\bUpi_1) -  m_{c}(D_i,\bp_1,\bUpi_2)\right|+
			\left|m_{c}(D_i,\bp_1,\bUpi_2) -  m_{c}(D_i,\bp_2,\bUpi_2)\right|\\
			&=\Bigg|
			\left(
			\left(\prod_{k=1}^{N_i}
			\sum_{a=0}^{1}1\{A_{ik}=a\}\pi_{1a}^{-1}(\bX_{ik})\right)
			-\left(\prod_{k=1}^{N_i}
			\sum_{a=0}^{1}1\{A_{ik}=a\}\pi_{2a}^{-1}(\bX_{ik})\right)    
			\right)\\
			&\hspace{1in}\times 
			\left\{ \partial_{p_c}\ell(\bp_1 |\widetilde{Z}_i, \bA_i) \right\}
			\Bigg|\\
			&\h2s+\Bigg|
			\left(\prod_{k=1}^{N_i}
			\sum_{a=0}^{1}1\{A_{ik}=a\}\pi_{2a}^{-1}(\bX_{ik})\right) \left\{ \partial_{p_c}\ell(\bp_1 |\widetilde{Z}_i, \bA_i) - \partial_{p_c}\ell(\bp_2 |\widetilde{Z}_i, \bA_i) \right\} \Bigg|\\
			&\leq  CN_i 2^{N_i}(\|\pi_{11}-\pi_{21}\|_{\infty}+\|\bp_1-\bp_2\|),
		\end{align*}
		for some constant $C$, where to get the last inequality we leverage the Lipschitz continuity of $\partial_{p_c}\ell(\cdot|\widetilde{Z}_i, \bA_i)$. Thus, condition (3.1) in \cite{chen2003estimation} is satisfied. Conditions (3.2) and (3.3)
		in \cite{chen2003estimation} hold by choosing $m_{lc}=0$ and Condition (2) of Theorem \ref{thm2}.
		Then condition (2.5') hold by Theorem 3.



		Now we check the condition (2.6').
		\begin{align*}
			&\h2s \Gamma_{2}(\bp, \bUpi^{*})[\widehat{\bUpi}-\bUpi^{*}] \\
			&=E\Bigg[
			\sum_{k=1}^{N}\Bigg[\Bigg(\sum_{a=0}^{1}1\{A_{k}=a\} \frac{(\pi_{a}^{*}-\widehat{\pi}_{a})}{\pi_{a}^{*2}}(\bX_{k})\Bigg)
			\prod_{l\neq k}\left(
			\sum_{a=0}^{1}1\{A_{l}=a\}\pi_{a}^{*-1}(\bX_{l})\right) 
			\Bigg]\\
			&\hspace{1in}\times    \left\{ 
			\partial_{p_c}\ell(\bp|\widetilde{Z}, \bA) \right\}
			\Bigg]\\
			&=E\Bigg[
			\sum_{k=1}^{N}\Bigg[\Bigg(
			1\{A_{k}=0\} \frac{\mb{E}_{M_{t}}\phi(\bX_{ij},\bX_{k})}{\pi_{0}^{*2}(\bX_{k})} -
			1\{A_{k}=1\} \frac{\mb{E}_{M_{t}}\phi(\bX_{ij},\bX_{k})}{\pi_{1}^{*2}(\bX_{k})}\Bigg)
			\\
			&\hspace{0.5in}\times   \prod_{l\neq k}\left(
			\sum_{a=0}^{1}1\{A_{l}=a\}\pi_{a}^{*-1}(\bX_{l})\right) 
			\Bigg] \left\{ 
			\partial_{p_c}\ell(\bp|\widetilde{Z}, \bA) \right\}
			\Bigg]\\
			& \hspace{1in}+o_{p}(M_{t}^{-1/2}) \\
			&=\mb{E}_{M_{t}}\Bigg\{E\Bigg[
			\sum_{k=1}^{N}\Bigg[\Bigg(
			1\{A_{k}=0\} \frac{\phi(\bX_{ij},\bX_{k})}{\pi_{0}^{*2}(\bX_{k})} -
			1\{A_{k}=1\} \frac{\phi(\bX_{ij},\bX_{k})}{\pi_{1}^{*2}(\bX_{k})}\Bigg)
			\\
			&\hspace{0.5in}\times   \prod_{l\neq k}\left(
			\sum_{a=0}^{1}1\{A_{l}=a\}\pi_{a}^{*-1}(\bX_{l})\right) 
			\Bigg] \left\{ 
			\partial_{p_c}\ell(\bp|\widetilde{Z}, \bA) \right\}
			\Bigg]\Bigg\} +o_{p}(M_{t}^{-1/2}) \\
			&=: \mb{E}_{M_t}(\psi(X_{ij}))+o_{p}(M_{t}^{-1/2}).
		\end{align*}
		We then obtain that 
		\begin{align*}
			&\h2s E[\psi(X_{ij})]\\
			&= E\Bigg[
			\sum_{k=1}^{N}\Bigg[\Bigg(
			-1\{A_{k}=1\} \frac{E\left[\phi(\bX_{ij},\bX_{k})|\bX_{k}\right]}{\pi_{1}^{*2}(\bX_{k})}
			+1\{A_{k}=0\} \frac{E\left[\phi(\bX_{ij},\bX_{k})|\bX_{k}\right]}{\pi_{0}^{*2}(\bX_{k})} \Bigg)
			\\
			&\hspace{0.2in}\times   \prod_{l\neq k}\left(
			\sum_{a=0}^{1}1\{A_{l}=a\}\pi_{a}^{*-1}(\bX_{l})\right) 
			\Bigg] \left\{ 
			\partial_{p_c}\ell(\bp|\widetilde{Z}, \bA) \right\}
			\Bigg]=0,
		\end{align*}
		since $E\left[\phi(\bX_{ij},\bX_{k})|\bX_{k}\right]=0$ as in Condition (1). Obviously,
		$E[\psi^2(X_{ij})]<\infty$. Therefore, condition (2.6') and consequently condition (2.6) in \cite{chen2003estimation} hold.  
		
		Applying Theorem 2 in \cite{chen2003estimation} yields that $\sqrt{M}(\widehat{\bp}-\bp^{*})\rightarrow_{d} N(0, \bUGamma)$. This completes the proof of Theorem \ref{thm2}.
		    
		\end{proof}

		\subsection{Lemmas}

        \begin{lemma}\label{lemma1}
			Under Condition (A\ref{Aex}), the terms
			$W(\ba, \mb{x}), \widetilde{\eta}_0(\ba;\bp), \widetilde{\phi}_1(\widetilde{\by}, \ba; \bp)$ are uniformly bounded and bounded away from $0$ for all $\widetilde{\by} \in \mc{Y}^N, \ba \in \mc{A}^N, \mb{x} \in \mc{X}^N$ and all $\bp \in [\nu, 1-\nu]^2$. Thus, for each fixed $\tilde{z}$ and $\ba$, the function $
            \ell(\bp|\tilde{z}, \ba) := (1 - \tilde{y}^{*})\log(\widetilde{\eta}_0(\ba;\bp)) + \tilde{y}^{*} \log(\widetilde{\phi}_1(\widetilde{\by}, \ba; \bp))
            $ is smooth for all $\bp \in [\nu, 1-\nu]^2$.
		\end{lemma}
        
		\begin{proof} Applying Condition (A\ref{Aex}) to the IPW (3) easily gives
        $
        1 < W(\ba, \mb{x}) = \prod_{j=1}^{N}
	\pi_{a_j}^{\ast -1}(\bx_{j}) < \nu^{-N}
        $ uniformly. Next, clearly $0 \leq \widetilde{\eta}_0(\ba;\bp), \widetilde{\phi}_1(\widetilde{\by}, \ba; \bp) \leq 1$, being probabilities, so we only have to show they are uniformly bounded away from $0$. Noting $\widetilde{\eta}_0(\ba;\bp) = \eta_0(\ba; \bp)S_p + \eta_1(\ba; \bp)(1-S_e)$ from \ref{rct_likelihood}, where each of the summands are themselves positive probabilities, we see that $0 < S_p \nu^N 
            \leq S_p \prod_{j=1}^N (1-p_{a_j}) = S_p \eta_0(\ba; \bp) \leq \widetilde{\eta}_0(\ba;\bp)$ uniformly. Similarly for each fixed $\widetilde{\by} \in \mc{Y}^N$, we can derive a strict bound $0 < b(\widetilde{\by}, S_e, S_p, \nu) \leq  \widetilde{\phi}_1(\widetilde{\by}, \ba; \bp)$ away from $0$ that is uniform over all $\ba \in \mc{A}^N$ and all $\bp \in [\nu, 1-\nu]^2$, but still depends on $\widetilde{\by}$. To eliminate this dependence, note that $\mc{Y}^N$ is a finite set. Thus, we can set the uniform lower bound as $0 < b(S_e, S_p, \nu) = \min_{\widetilde{\by} \in \mc{Y}^N} b(\widetilde{\by}, S_e, S_p, \nu)$.

            Finally, for fixed $\tilde{z}$ and $\ba$, $\widetilde{\eta}_0(\ba;\bp), \widetilde{\phi}_1(\widetilde{\by}, \ba; \bp)$ are polynomials in $\bp$. So, the derivatives of $\ell(\bp|\tilde{z}, \ba)$ are rational functions whose denominators are some powers of $\widetilde{\eta}_0(\ba;\bp), \widetilde{\phi}_1(\widetilde{\by}, \ba; \bp)$. As a result, they exist everywhere because the denominators are bounded away from $0$.

		\end{proof}

		\begin{lemma}\label{lemma2}
        For $b,c=0,1$, define the parametrized function classes
		\begin{gather*}
			\mc{F}_{b} := \Bigg\{f_{b,\bp}: \mc{Z} \times \mc{A}^N \times \mc{X}^N \to \mb{R} \ \Big|\ f_{b,\bp}(\tilde{z}, \ba, \mb{x}) = W(\ba, \mb{x}) \partial_{p_b}\ell(\bp|\tilde{z}, \ba), \bp\in [\nu, 1-\nu]^2\Bigg\},   \\
			\mc{G}_{bc} := \Bigg\{g_{bc,\bp}: \mc{Z} \times \mc{A}^N \times \mc{X}^N \to \mb{R}  \ \Big | \ g_{bc,\bp}(\tilde{z}, \ba, \mb{x}) = W(\ba, \mb{x}) \partial^2_{p_c p_b}\ell(\bp|\tilde{z}, \ba), \bp\in [\nu, 1-\nu]^2 \Bigg\},
		\end{gather*} Under Condition (A\ref{Aex}),
			$\mc{F}_{b}, b=0,1$'s and $\mc{G}_{bc}, b,c=0,1$'s are Donsker classes.
		\end{lemma}
		
         \begin{proof} By Lemma \ref{lemma1}, for each fixed $\tilde{z}, \ba, \mb{x}$, $\nabla_{\bp} f_{b, \bp}(\tilde{z}, \ba, \mb{x})$ is continuous and bounded on the compact parameter box $[\nu, 1-\nu]^2$. So it is Lipschitz in the parameters viz
        $$
        |f_{b, \bp_1}(\tilde{z}, \ba, \mb{x}) - f_{b, \bp_2}(\tilde{z}, \ba, \mb{x})| \leq m(\tilde{z}, \ba, \mb{x}) \|\bp_1 - \bp_2\|
        $$ where $m(\tilde{z}, \ba, \mb{x}) = \sup_{\bp\in [\nu, 1-\nu]^2} \|\nabla_{\bp} f_{b, \bp}(\tilde{z}, \ba, \mb{x})\|$.
		
		Then by Theorem 2.7.17 of \cite{vaart2023weak},  for any probability measure $Q$, the bracketing numbers of the class $\mc{F}_{b}$ are bounded by the covering numbers of the compact parameter box $[\nu, 1-\nu]^2$
		\begin{align*}
			N_{[]}(2 \epsilon \|m\|_{Q,2}, \mc{F}_b, L_2(Q)) \leq N(\epsilon, [\nu, 1-\nu]^2, \|\cdot\|_2)
		\end{align*}
		for  $\epsilon > 0$. For the definitions of covering and bracketing numbers, we refer the reader to Definitions 2.1.5 and 2.1.6 on page 132 of \cite{vaart2023weak}. Now, since $N(\epsilon, [\nu, 1-\nu]^2, \|\cdot\|_2) = \begin{cases} \lceil K (1/\eps)^2 \rceil &\text{if } 0 < \epsilon < \text{diam}([\nu, 1-\nu]^2)  \\ 1 & \text{if } \text{diam}([\nu, 1-\nu]^2) < \epsilon < \infty  \end{cases} $ for some scaling constant $K > 0$
		\begin{align*}
			\int_{0}^{\infty}\sqrt{\log N_{[]}(2 \epsilon \|m\|_{Q,2}, \mc{F}_b, L_2(Q))}d\epsilon   \leq  &\int_{0}^{\infty} 
			\sqrt{\log N(\epsilon, [\nu, 1-\nu]^2, \|\cdot\|_2)}d\epsilon \\ = & \int_{0}^{\text{diam}([\nu, 1-\nu]^2)} 
			\sqrt{\log \left( \lceil K(1/\epsilon)^2 \rceil \right) }d\epsilon <\infty.
		\end{align*}
		So by Theorem 2.5.6 of \cite{vaart2023weak},
		$\mc{F}_{b}, b=0,1$'s are Donsker. Similarly, we can show that $\mc{G}_{bc}, b,c=0,1$'s are Donsker as well. This completes the proof of Lemma \ref{lemma2}.

        \end{proof}

		\begin{lemma}\label{lemma3}
			Under Condition  (A\ref{Aex}) and the  Rubin Causal Model assumptions (RCM1) -- (RCM2),
			$E[\bU(\bp^\ast; \bD)] = 0$.
		\end{lemma}
		
         \begin{proof} Since $\bU(\bp^\ast; \bD) = \frac{1}{M}\sum_{i=1}^M \bU(\bp^\ast; D_i)$, where $\bU(\bp^\ast; D_i) = W(\bA_i, \mb{X}_i) \nabla \ell(\bp^\ast | \widetilde{Z}_i, \bA_i)$, it suffices to show $E[\bU(\bp^\ast; D_i)]=0$ for each $i$-th group data $D_i = (\widetilde{Z}_i, \bA_i, \mb{X}_i)$. 
        \begin{align*}
			&E[\bU(\bp^\ast; D_i)] 
			= \sum_{\ba_i \in \mc{A}^{N_i}} \sum_{\tilde{z}_i \in \mc{Z}_i} \int_{\mb{x}_i \in \mc{X}^{N_i}} W(\ba_i, \mb{x}_i) \nabla \ell(\bp^\ast | \tilde{z}_i , \ba_i) \pi_{\widetilde{Z}_i , \bA_i, \mb{X}_i}(\tilde{z}_i , \ba_i, \mb{x}_i ; \bp^\ast) d\mb{x}_i  \\
			=&\sum_{\ba_i \in \mc{A}^{N_i}} \sum_{\tilde{z}_i \in \mc{Z}_i} \int_{\mb{x}_i \in \mc{X}^{N_i}} \frac{\nabla_{\bp} \pi_{\widetilde{Z}_{i}^{(\ba_i)}}(\tilde{z}_{i} ; \bp^\ast) }{ \pi_{\bA_i | \mb{X}_i}(\ba_i | \mb{x}_i) \pi_{\widetilde{Z}_{i}^{(\ba_i)}}(\tilde{z}_{i} ; \bp^\ast) } \pi_{\widetilde{Z}_i | \bA_i, \mb{X}_i}(\tilde{z}_i | \ba_i, \mb{x}_i ; \bp^\ast) \pi_{\bA_i | \mb{X}_i}(\ba_i | \mb{x}_i) f_{\mb{X}_i}(\mb{x}_i)  d\mb{x}_i \\
			=&\sum_{\ba_i \in \mc{A}^{N_i}} \sum_{\tilde{z}_i \in \mc{Z}_i}  \left( \int_{\mb{x}_i \in \mc{X}^{N_i}}\pi_{\widetilde{Z}_i | \bA_i, \mb{X}_i}(\tilde{z}_i | \ba_i, \mb{x}_i ; \bp^\ast) f_{\mb{X}_i}(\mb{x}_i)  d\mb{x}_i  \right) \frac{\nabla_{\bp} \pi_{\widetilde{Z}_{i}^{(\ba_i)}}(\tilde{z}_{i} ; \bp^\ast) }{ \pi_{\widetilde{Z}_{i}^{(\ba_i)}}(\tilde{z}_{i} ; \bp^\ast) } \\
			=&\sum_{\ba_i \in \mc{A}^{N_i}} \sum_{\tilde{z}_i \in \mc{Z}_i}  \left( \int_{\mb{x}_i \in \mc{X}^{N_i}}\pi_{\widetilde{Z}_{i}^{(\ba_i)} | \bA_i, \mb{X}_i}(\tilde{z}_i | \ba_i, \mb{x}_i ; \bp^\ast) f_{\mb{X}_i}(\mb{x}_i)  d\mb{x}_i  \right) \frac{\nabla_{\bp} \pi_{\widetilde{Z}_{i}^{(\ba_i)}}(\tilde{z}_{i} ; \bp^\ast) }{ \pi_{\widetilde{Z}_{i}^{(\ba_i)}}(\tilde{z}_{i} ; \bp^\ast) } \quad \text{(consistency)} \\
			=&\sum_{\ba_i \in \mc{A}^{N_i}} \sum_{\tilde{z}_i \in \mc{Z}_i}  \left( \int_{\mb{x}_i \in \mc{X}^{N_i}}\pi_{\widetilde{Z}_{i}^{(\ba_i)} |  \mb{X}_i}(\tilde{z}_i | \mb{x}_i ; \bp^\ast) f_{\mb{X}_i}(\mb{x}_i)  d\mb{x}_i  \right) \frac{\nabla_{\bp} \pi_{\widetilde{Z}_{i}^{(\ba_i)}}(\tilde{z}_{i} ; \bp^\ast) }{ \pi_{\widetilde{Z}_{i}^{(\ba_i)}}(\tilde{z}_{i} ; \bp^\ast) } \quad \text{(exchangeability)} \\
			=&\sum_{\ba_i \in \mc{A}^{N_i}} \sum_{\tilde{z}_i \in \mc{Z}_i}   \pi_{\widetilde{Z}_{i}^{(\ba_i)}}(\tilde{z}_{i} ; \bp^\ast) \frac{\nabla_{\bp} \pi_{\widetilde{Z}_{i}^{(\ba_i)}}(\tilde{z}_{i} ; \bp^\ast) }{ \pi_{\widetilde{Z}_{i}^{(\ba_i)}}(\tilde{z}_{i} ; \bp^\ast) } = \sum_{\ba_i \in \mc{A}^{N_i}} \nabla_{\bp} \left (\sum_{\tilde{z}_i \in \mc{Z}_i}   \pi_{\widetilde{Z}_{i}^{(\ba_i)}}(\tilde{z}_{i} ; \bp^\ast) \right) = 0 
		\end{align*}

         \end{proof}

		\begin{lemma}\label{lemma4}
			Under Conditions (A\ref{Aex}) -- (A\ref{Apd}) and the Rubin Causal Model assumptions (RCM1) -- (RCM2),
			$-E[\dot{\bU}(\bp^\ast; \bD)]$ and $E[\bU(\bp^\ast; \bD)\bU(\bp^\ast; \bD)^{\rT}]$ are positive definite.
		\end{lemma}
		
         \begin{proof} It suffices to show $-E[\dot{\bU}(\bp^\ast; D_i)]$ is positive definite as positive sums of positive definite matrices are positive definite. So a similar derivation as Lemma \ref{lemma3} for
        \begin{align*}
            &\dot{\bU}(\bp^\ast; D_i) = W(\bA_i, \mb{X}_i) \sum_{\ba_i \in \mc{A}^{N_i}} 1\{\bA_i = \ba_i\} \nabla^2 \ell(\bp^\ast | \widetilde{Z}_i, \bA_i) \\
             = &W(\bA_i, \mb{X}_i) \sum_{\ba_i \in \mc{A}^{N_i}} 1\{\bA_i = \ba_i\} \Bigg( \frac{\nabla_{\bp}^2 \pi_{\widetilde{Z}_{i}^{(\ba_i)}}(\widetilde{Z}_{i} ; \bp^\ast) }{\pi_{\widetilde{Z}_{i}^{(\ba_i)}}(\widetilde{Z}_{i} ; \bp^\ast)} 
            - \frac{\nabla_{\bp} \pi_{\widetilde{Z}_{i}^{(\ba_i)}}(\widetilde{Z}_{i} ; \bp^\ast) \nabla_{\bp}^{\rT} \pi_{\widetilde{Z}_{i}^{(\ba_i)}}(\widetilde{Z}_{i} ; \bp^\ast)}{ \pi_{\widetilde{Z}_{i}^{(\ba_i)}}(\widetilde{Z}_{i} ; \bp^\ast)^2} \Bigg)
        \end{align*} produces
        \begin{align*}
            -E[\dot{\bU}(\bp^\ast; D_i)] = &\sum_{\ba_i \in \mc{A}^{N_i}} \sum_{\tilde{z}_i \in \mc{Z}_i} \left(\frac{1}{ \pi_{\widetilde{Z}_{i}^{(\ba_i)}}(\tilde{z}_{i} ; \bp^\ast)} \right) \nabla_{\bp} \pi_{\widetilde{Z}_{i}^{(\ba_i)}}(\tilde{z}_{i} ; \bp^\ast) \nabla_{\bp}^{\rT} \pi_{\widetilde{Z}_{i}^{(\ba_i)}}(\tilde{z}_{i} ; \bp^\ast) \\
            = &\sum_{\ba_i \in \mc{A}^{N_i}} \sum_{\tilde{z}_i \in \mc{Z}_i}  (\kappa_{\tilde{z}_{i}, \ba_i, \bp^\ast})  \nabla_{\bp} \pi_{\widetilde{Z}_{i}^{(\ba_i)}}(\tilde{z}_{i} ; \bp^\ast) \nabla_{\bp}^{\rT} \pi_{\widetilde{Z}_{i}^{(\ba_i)}}(\tilde{z}_{i} ; \bp^\ast).
        \end{align*} This is evidently a positively weighted sum of vector outer products, which is positive definite if and only if those vectors span their inhabitant space viz Condition (A\ref{Apd}).
        
        Next, note first that
        \begin{align*}
            \bU(\bp^\ast; \bD)\bU(\bp^\ast; \bD)^{\rT} = \frac{1}{M^2} \left( \sum_{i = 1}^M \bU(\bp^\ast; D_i)\bU(\bp^\ast; D_i)^{\rT}   
            +  \sum_{j, k = 1; j \neq k}^M \bU(\bp^\ast; D_j)\bU(\bp^\ast; D_k)^{\rT} \right)
        \end{align*} Since the data $D_j$ is independent of $D_k$ if $j \neq k$, the expectations of the cross terms split and are therefore zero by Lemma \ref{lemma3}. Now, we can follow a similar derivation as Lemma \ref{lemma3} for the first term to get
        \begin{align*}
            E[\bU(\bp^\ast; D_i)\bU(\bp^\ast; D_i)^{\rT}] = &\sum_{\ba_i \in \mc{A}^{N_i}} \sum_{\tilde{z}_i \in \mc{Z}_i} \left(\frac{1}{ \pi_{\widetilde{Z}_{i}^{(\ba_i)}}(\tilde{z}_{i} ; \bp^\ast)} \right)^2 \left( \int_{\mb{x}_i \in \mc{X}^{N_i}} \frac{\pi_{\widetilde{Z}_i^{(\ba_i)} | \mb{X}_i}(\tilde{z}_i | \mb{x}_i ; \bp^\ast)}{\pi_{\bA_i | \mb{X}_i}(\ba_i | \mb{x}_i)} f_{\mb{X}_i}(\mb{x}_i) d \mb{x}_i \right) \\
            &\h2s\h2s\h2s\h2s \times\nabla_{\bp} \pi_{\widetilde{Z}_{i}^{(\ba_i)}}(\tilde{z}_{i} ; \bp^\ast) \nabla_{\bp}^{\rT} \pi_{\widetilde{Z}_{i}^{(\ba_i)}}(\tilde{z}_{i} ; \bp^\ast) \\
            =&\sum_{\ba_i \in \mc{A}^{N_i}} \sum_{\tilde{z}_i \in \mc{Z}_i}  (\kappa_{\tilde{z}_{i}, \ba_i, \bp^\ast})^2 (\gamma_{\tilde{z}_{i}, \ba_i, \bp^\ast})  \nabla_{\bp} \pi_{\widetilde{Z}_{i}^{(\ba_i)}}(\tilde{z}_{i} ; \bp^\ast) \nabla_{\bp}^{\rT} \pi_{\widetilde{Z}_{i}^{(\ba_i)}}(\tilde{z}_{i} ; \bp^\ast).
        \end{align*} This is again a positively weighted sum of vector outer products. Note that the weights are finite because
        $
        \gamma_{\tilde{z}_{i}, \ba_i, \bp^\ast}$ $=$ $\int_{\mb{x}_i \in \mc{X}^{N_i}} \frac{\pi_{\widetilde{Z}_i^{(\ba_i)} | \mb{X}_i}(\tilde{z}_i | \mb{x}_i ; \bp^\ast)}{\pi_{\bA_i | \mb{X}_i}(\ba_i | \mb{x}_i)} f_{\mb{X}_i}(\mb{x}_i) d \mb{x}_i$ $< \infty
        $ as $\pi_{\bA_i | \mb{X}_i}(\ba_i | \mb{x}_i)$ is uniformly bounded away from $0$ by Condition (A\ref{Aex}). Hence, Condition (A\ref{Apd}) guarantees positive definiteness of the expression analogously.

        \end{proof}

        \begin{coro}\label{corollary4}
			Under Conditions (A\ref{Aex}) -- (A\ref{Apd}) and Rubin Causal Model assumptions (RCM1) -- (RCM2), there exists a constant $R_+ > 0$, independent of $M$, such that
			$-E[\dot{\bU}(\bp; \bD)]$ and $E[\bU(\bp; \bD)\bU(\bp; \bD)^{\rT}]$ are positive definite for all $\bp \in B_{R_+}(\bp^\ast) \cap [\nu, 1 - \nu]^2$.
		\end{coro}
        
       \begin{proof} By Lemma \ref{lemma1}  $E[\dot{\bU}(\bp; D_1)], E[\bU(\bp; D_1)\bU(\bp; D_1)^{\rT}]$ are continuous on $[\nu, 1 - \nu]^2$. Thus as $\text{det}\{-E[\dot{\bU}(\bp^\ast; D_1)]\} > 0$ and $\text{det}\{E[\bU(\bp^\ast; D_1)\bU(\bp^\ast; D_1)^{\rT}]\} > 0$ by Lemma \ref{lemma4}, there exists a ball $B_{R_+}(\bp^\ast)$ of some radius $R_+ > 0$ such that $\text{det}\{-E[\dot{\bU}(\bp; D_1)]\} > 0$ and $\text{det}\{E[\bU(\bp; D_1)\bU(\bp; D_1)^{\rT}]\} > 0$ for all $\bp \in B_{R_+}(\bp^\ast) \cap [\nu, 1 - \nu]^2$. Now note that $E[\dot{\bU}(\bp; \bD)] = E[\dot{\bU}(\bp; D_1)]$ and $E[\bU(\bp; \bD)\bU(\bp; \bD)^{\rT}] = \frac{1}{M}E[\bU(\bp; D_1)\bU(\bp; D_1)^{\rT}]$. 

       \end{proof}

\end{appendices}

\end{document}